\documentclass{amsart}

\usepackage{amsmath,amsthm}
\usepackage{amssymb}
\usepackage{color}

\usepackage{rotating}
\usepackage{algorithm}

\usepackage{algorithmic}

\newtheorem{theorem}{Theorem}[section]
\newtheorem{proposition}[theorem]{Proposition}
\newtheorem{corollary}[theorem]{Corollary}
\newtheorem{lemma}[theorem]{Lemma}
\theoremstyle{definition}
\newtheorem{definition}[theorem]{Definition}

\newtheorem{example}[theorem]{Example}
\theoremstyle{remark}
\newtheorem{remark}[theorem]{Remark}

\newcommand{\field}[1][]{\mathbb{F}_{#1}}
\DeclareMathOperator{\Galois}{Gal}
\DeclareMathOperator{\weight}{w}
\DeclareMathOperator{\distance}{d}
\newcommand{\Aut}[2][]{\operatorname{Aut}_{#1}(#2)}
\newcommand{\norm}[2]{N_{#1}(#2)}
\DeclareMathOperator{\rank}{rk}
\newcommand{\lclm}[1]{\left[#1\right]_\ell}

\newcommand{\gcrd}[1]{\left(#1\right)_r}

\newcommand{\matrixring}[2]{\mathcal{M}_{#1}(#2)}
\newcommand{\tovector}{\mathfrak{v}}

\newcommand{\pseudobound}[2][]{\overline{#2}^{#1}}
\DeclareMathOperator{\cs}{cs}
\newcommand{\Lq}[2][k]{\mathcal{L}_{q^{#1}}\field[q^{#2}][t]}

\begin{document}

\title[Hartmann--Tzeng bound and Skew Cyclic Codes]{Hartmann--Tzeng bound and Skew Cyclic Codes of Designed Hamming Distance}
\author[J. G\'{o}mez-Torrecillas]{Jos\'e G\'omez-Torrecillas}
\author{F. J. Lobillo}
\author[G. Navarro]{Gabriel Navarro}
\author[A. Neri]{Alessando Neri}
\thanks{J. G\'{o}mez-Torrecillas and F. J. Lobillo are with CITIC and Department of Algebra, University of Granada.}
\thanks{G. Navarro is with CITIC and Department of Computer Sciences and AI, University of Granada.}
\thanks{A. Neri is with Institut f\"{u}r Mathematik, University of Z\"{u}rich.}
\thanks{Research partially supported by grants MTM2013-41992-P and TIN2013-41990-R from Ministerio de Econom\'{\i}a y Competitividad of the Spanish Government and from FEDER, by grant MTM2016-78364-P from Agencia Espa\~nola de Investigaci\'{o}n and from FEDER, and by grant number 169510 from the Swiss National Science Foundation.}

\begin{abstract}
The use of skew polynomial rings allows to endow linear codes with cyclic structures which are not cyclic in the classical (commutative) sense. Whenever these skew cyclic structures are carefully chosen, some control over the Hamming distance is gained, and it is possible to design efficient decoding algorithms. In this paper, we give a version of the Hartmann-Tzeng bound that works for a wide class of skew cyclic codes. We also provide a practical method for constructing them with designed distance. For skew BCH codes, which are covered by our constructions, we discuss decoding algorithms. Detailed examples illustrate both the theory  as the constructive methods it supports.
\end{abstract}

\keywords{Linear codes, convolutional codes, cyclic codes, skew cyclic codes, Hartmann-Tzeng bound, BCH skew code}

\maketitle

\section{Introduction}

The availability of additional algebraic structure in error correcting linear codes has helped their construction in two ways. On one hand, a better knowledge of the main parameters of the code can be obtained. For instance the pioneer works \cite{Bose/Chaudhuri:1960,Hocquenghem:1959} use cyclic structures in linear block codes to design them with a prescribed distance. Those ideas were generalized by Hartmann and Tzeng in \cite{Hartmann/Tzeng:1972}. In these papers, the behavior of the roots of the cyclic generator in a suitable field extension of the alphabet (a finite field), provides lower bounds of the Hamming distance of the code. On the other hand, the presence of higher algebraic structures also allows the design of fast decoding algorithms. This is the case for instance of the Peterson--Gorenstein--Zierler algorithm, see \cite{Peterson:1960,Gorenstein/Zierler:1961}, where linear algebra techniques are used, in conjunction with the aforementioned behavior of the roots of cyclic codes, to design decoding procedures. 

Let $\field{}$ be a finite field. Classically, a  cyclic block code over the alphabet $\field{}$ is an ideal of \(\field{}[x]/\langle f \rangle \cong \field{}^n\), where $f$ is a polynomial of degree $n$. In this way, only a few codes of length \(n\) enjoy a cyclic structure modeled by $f$ (in fact they are in one to one correspondence with the monic divisors of \(f\)) and most of the codes are not cyclic. However, by replacing \(\field{}[x]/\langle f \rangle\) by some \(n\)-dimensional non-commutative \(\field{}\)-algebra, new cyclic structures on some non cyclic codes arise. One of the most successful ways to follow this philosophy consists in twisting the multiplication of the polynomial ring. Concretely, skew cyclic block codes are left ideals of factor algebras of skew polynomial rings \(\field{}[x;\sigma]\) by a two-sided ideal, where \(\sigma\) is an automorphism of the finite field \(\field\). Such an ideal is generated by a normal polynomial $f \in \field{}[x;\sigma]$, so that skew cyclic codes will be  in correspondence with right divisors of $f$. It is well known that this number of divisors is much larger than in the commutative case, due essentially to the lack of uniqueness, in the usual sense, of factorizations in \(\field{}[x;\sigma]\) (see \cite{Ore:1933}). Of course we want to get some control on the parameters of the skew cyclic code, and also to take advantage of this cyclic structure to design efficient decoding algorithms. To this end, both \(\sigma\) and \(f\) have to be carefully chosen. These skew block codes were introduced, for $f = x^n -1$, in \cite{Boucher/etal:2007}, and, in the general case, independently in \cite{Boucher/Ulmer:2009,Gabidulin:2009}. The notion can be traced back to \cite{Gabidulin:1985}, where the author used the arithmetics of linearized polynomials to introduce and investigate the nowadays known as Gabidulin codes.

In \cite{Boucher/etal:2007} bounds of the Hamming distance and a Sugiyama like decoding algorithm are provided for skew cyclic codes when the alphabet is \(\field[2^n]\), the automorphism is the Frobenius automorphism, \(\sigma(a) = a^2\), and \(f = x^n-1\). One advantage of this choice of parameters is that \(f\) fully decomposes as a least common left multiple of \(n\) linear polynomials. For a more general \(f\) a way to find a decomposition of this type is needed. In \cite{Chaussade/etal:2009} Picard--Vessiot fields of \(\sigma\)-difference equations associated to skew codes are used to find those non-commutative roots and design skew BCH codes with prescribed distance; a new Sugiyama like decoding algorithm, extending the one in \cite{Boucher/etal:2007}, is also proposed, based on the resolution of a commutative key equation in a suitable commutative polynomial ring. In \cite{GLNPGZ}, a decoding algorithm inspired by the classical Peterson-Gorenstein-Zierler algorithm is designed for skew Reed-Solomon codes (skew RS codes), taking \(f = x^n - 1\), where \(n\) is the order of the automorphism \(\sigma\) over a finite field \(\field{}\). When \(n\) is not the order of \(\sigma\), it is possible to add non commutative roots to \(f\) extending \(\sigma\) to a suitable field extension of $\field{}$. In this paper, we use these new roots to provide bounds on the Hamming distance similar to \cite{Hartmann/Tzeng:1972}.

Our approach was motivated by the new perspective of cyclicity in convolutional codes proposed in \cite{gln2016new}. A naive attempt to provide cyclic structures on convolutional codes fails as pointed out in \cite{Piret:1975}, and the use of skew polynomials was immediately proposed \cite{Piret:1976,Roos:1979} and further developed in \cite{Gluesing/Schmale:2004,GLN2017IdealCodes}. Initial approaches to an algebraic decoding in this setting have been proposed in \cite{gln2017sugiyama} and \cite{GLNPGZ}. In \cite{gln2017sugiyama}, a pure non commutative key equation is presented and solved. In \cite{GLNPGZ}, the aforementioned non-commutative Peterson-Gorenstein-Zierler like decoding algorithm is designed by using abstract fields, covering both block and convolutional skew RS codes. Note that skew RS block codes in the sense of \cite{GLNPGZ} can be regarded as special cases of the GSRS codes introduced in \cite[Definition 9]{Liu/etal:2015}, see Remark \ref{RSGSRS}. Skew RS block codes were also proved to be MDS in \cite[Theorem 5]{Boucher/Ulmer:2014}. Therefore, the decoding algorithms designed in \cite{Boucher/Ulmer:2014} and \cite{Liu/etal:2015} are also available for skew RS block codes.

In this paper we continue with the study of skew cyclic structures in linear codes. In order to apply our results to both cases, block codes and convolutional codes, we work in the framework of a general finite field extension, well understood that our motivating examples are finite fields and rational functions over finite fields. In a sense, some of our results might be regarded as extensions of existing ones for block codes in the literature to convolutional codes and, more generally, to finite field extensions with cyclic Galois group. An Appendix, describing which results, and how,  are concerned is included.

In Section \ref{algsetup} the mathematical tools needed to handle non-commutative roots of skew polynomials are developed. This kind of roots of \(x^n-1\) is used in Sections \ref{HT} and \ref{codesconstruction} to give a Hartmann-Tzeng-like bound of the Hamming distance of skew cyclic codes and to build codes with designed distance. Theorem \ref{HTbound} may be seen as an extension of  \cite[Corollary 5]{MartinezPenas:2017} to a general context, which includes skew convolutional codes, and, similarly, Corollary \ref{BCHbound} can be considered as an extension of \cite[Proposition 1]{Chaussade/etal:2009}.  However, our approach and methods are conceptually different from that of \cite{MartinezPenas:2017}, based upon root spaces, or  from the Picard-Vessiot view used in \cite{Chaussade/etal:2009}. Our method of constructing skew codes of prescribed parameters is based on a very simple combinatorial procedure, namely, compute the closure of suitable subsets of $\{0,1, \dots, n-1 \}$ under the addition of a modulus $s$, see Theorem \ref{Hdist}. Such a method is new, up to our best knowledge, even for skew block codes.

Finally we show in Section \ref{sec:decoding} that the non commutative Peterson-Gorenstein-Zierler algorithm designed in \cite{GLNPGZ} and the Sugiyama's like algorithm which appears in \cite{Boucher/etal:2007,gln2017sugiyama} to decode skew RS codes, can be used to also decode the skew BCH codes introduced here. In the Appendix, we also include a brief discussion of how the decoding algorithms from \cite{Augot/etal:2013} and \cite{Liu/etal:2015} are applicable to skew RS codes and, hence, to skew BCH codes.

We include explicit examples illustrating the theory. All computations have been done with the aid of SageMath \cite{sagemath}.

\section{Skew cyclic codes}\label{algsetup}

Let \(L\) be a field and let \(\mathcal{C} \subseteq L^n\) be an \(L\)--linear code (i.e. a non zero vector subspace of \(L^n\)). In order to endow \(\mathcal{C}\) with a skew cyclic structure we need to introduce some notation. We denote by $\Aut{L}$ the automorphism group of \(L\). From each $\sigma \in \Aut{L}$ we can construct the \emph{skew polynomial ring} $R = L[x;\sigma]$.  Its elements are polynomials in one variable $x$ with coefficients in $L$ written on the left. The sum of polynomials is as usual, while the multiplication is twisted according to the rule $xa = \sigma(a)x$, for all $a \in L$. This construction is a particular case of the most general notion of Ore polynomial extension, where a $\sigma$-derivation is also considered, see \cite{Ore:1933}. It is well-known that $R$ has left and right Euclidean Division Algorithms, so, in particular, every left and every right ideal of $R$ is principal (see, e.g., \cite{Jacobson:1996}). In particular, if $\{ f_1, \dots, f_k \} \subseteq R$, the left ideals $Rf_1 \cap \cdots \cap Rf_k$ and \(Rf_1 + \cdots + Rf_k\) are principal. Their monic generators are the \emph{least common left multiple} and the \emph{greatest common right divisor} of $f_1, \dots, f_k$ respectively, denoted by $\lclm{f_1, \dots, f_k}$ and \(\gcrd{f_1, \dots, f_k}\). These polynomials can be computed by means of the extended left Euclidean Algorithm (see \cite[Ch. I, Theorem 4.33]{Bueso/alt:2003}). Given polynomials $h, g \in R$, we use the notation $g \mid_r h$ to declare that $g$ is a \emph{right divisor of} $h$, that is, $Rh \subseteq Rg$. Moreover, \(\gamma \in L\) is said to be a \emph{right root} of \(g \in R\) if \(x-\gamma \mid_r g\). 

Now, assume that $\sigma$ has finite order $|\sigma| = \mu$, with $\mu \mid n$, and set $K = L^\sigma$, the invariant subfield under $\sigma$. Since \(\mu \mid n\), the polynomial \(x^n-1\) is a central polynomial in \(R\), hence \(\mathcal{R} = R/R(x^n-1)\) is a factor algebra over $K$.

As usual in algebraic coding theory, we identify elements of $\mathcal{R}$ with polynomials in $R$ of degree strictly lower than $n$.  The multiplication \emph{on the left} of elements of $R$ by elements of $L$ endows $R$ with the structure of an $L$--vector space. Hence,  $\mathcal{R}$ becomes an $L$--vector space isomorphic to $L^n$ via the coordinate map \(\tovector:\mathcal{R} \to L^n\), which maps each polynomial of degree lower than \(n\) to the vector of its coefficients. This identification is made all along the paper. The  Hamming weight of $u \in L^n$ is denoted by $\weight(u)$.

\begin{definition}\label{scc}
We say that \(\mathcal{C}\) is a \emph{skew cyclic code over \(L\) of length \(n\)} if \(\mathcal{C} = \tovector(\mathcal{R}g) \), where \(g \in R\) is a proper right divisor of \(x^n-1\). The \emph{dimension} of the code $\mathcal{C}$ is \(n-\deg g\), and its \emph{minimum Hamming distance} is
\[
\distance(\mathcal{C}):= \min\{\weight(c)\mid c\in \mathcal C, c\neq 0\}.
\]
\end{definition}

The structure of skew cyclic codes is linked to right divisors of \(x^n-1\). These divisors can be better understood if we can fully decompose $x^n-1$ as a least common left multiple of linear polynomials, i.e. \(x^n-1 = \lclm{x-\gamma_0, \dots, x-\gamma_{n-1}}\) for suitable right roots \(\gamma_0, \dots, \gamma_{n-1}\). These roots could exist or not in $L$, but one way to look for them is to extend \(L\) to a bigger field. The automorphism \(\sigma\) must also be extended in order to obtain compatibility with the skew polynomial ring structure. Next definition establishes a setting for this extension. 

\begin{definition}\label{autoextension}
Let $s$ be a positive integer. We say that $\sigma$ \emph{has an extension $\theta$ of degree $s$} if  there exists a field extension \(L \subseteq M\) and \(\theta \in \Aut{M}\) such that \(|\theta| = n = s\mu\), \(\theta_{|L} = \sigma\) and \(M^\theta = L^\sigma = K\). 
\end{definition}

So, \(L/K\),  \(M/K\) and $M/L$ are Galois extensions with Galois groups \(\Galois(L/K) = \langle \sigma \rangle\), \(\Galois(M/K) = \langle \theta \rangle\), and  \(\Galois(M/L) = \langle \pi \rangle\) where \(\pi = \theta^\mu\). We keep these assumptions and notation along the paper. 

Recall that \(R = L[x;\sigma]\) and let \(S = M[x;\theta]\). It follows from \(\theta_{|L} = \sigma\) that \(R \subseteq S\). Since  \(\Galois(M/K)\) is a cyclic group generated by $\theta$, every $\varrho \in \Galois(M/K)$ can be extended to a ring automorphism of \(S = M[x;\theta]\) by
\[
\begin{split}
\varrho : M[x;\theta] &\longrightarrow M[x;\theta] \\
f = \textstyle\sum_i a_i x^i &\longmapsto \varrho(f) = f^\varrho = \textstyle\sum_i \varrho(a_i) x^i.
\end{split}
\]

Obviously, we still have $\theta_{|R} = \sigma$.  On the other hand,  
\begin{equation}\label{Rpi}
R = \{ f \in S : f^{\pi} = f \}. 
\end{equation}

\begin{lemma}\label{l31}
For any \(f\in R\), \(Sf \cap R = Rf\).
\end{lemma}

\begin{proof}
It is clear that \(Rf \subseteq Sf \cap R\). Let $g \in S$ such that $ gf \in R$. By \eqref{Rpi}, 
$ gf = (gf)^\pi=g^\pi f^\pi=g^\pi f\), so \(g=g^\pi\) because \(S\) is a domain. Therefore, \(g \in R\), and $gf \in Rf$, which proves that $Sf \cap R \subseteq Rf$. 
\end{proof}

Let \(\mathcal{S} = S / S(x^n-1)\). By Lemma \ref{l31}, \(R(x^n-1) = R \cap S(x^n-1)\), hence there is a canonical inclusion \(\mathcal{R} \subseteq \mathcal{S}\). Analogously to \cite[Theorem 1]{gln2016new}, we get that $\mathcal{S}$ is isomorphic to the $n \times n$ matrix $K$--algebra  $\matrixring{n}{K}\), which shows that $x^n-1$ fully decomposes as least common left multiple of linear polynomials in $S$. Indeed, each of these decompositions corresponds to an expression of the zero ideal of $\mathcal{S}$ as an intersection of maximal left ideals.

In order to apply the former construction both to skew cyclic block and convolutional codes, we finish this section by analyzing how to construct automorphism extensions in the sense of Definition \ref{autoextension} in the cases $L=\field[q^{\mu}]$, a finite field, or $L = \field[](z)$, the rational function field over a finite field $\field[]$.

\begin{example}[\emph{Skew cyclic block codes}]\label{ex:scbc}
Let $M$ be a finite field, and $\theta \in \Aut{M}$ with $|\theta| = n = \mu s$. If $M^\theta = \field[q]$, the field of $q$ elements, then $M = \field[q^{\mu s }]$, $\theta = \tau^k$ where $\tau$ is the $\field[q]$--automorphism of Frobenius, and $(k,\mu s) = 1$, where, all along the paper, $(-,-)$ denotes the greatest common divisor of two natural numbers. Now, let $k = a\mu + h$ the Euclidean Division of $k$ by $\mu$. Then $L = M^{\theta^s} = \field[q^\mu]$, and $\sigma = \tau^h$ is an automorphism of order $\mu$ of $\field[q^\mu]$ such that $\theta_{|L} = \sigma$. Moreover, $L^{\sigma} = \field[q]$, since $(h,\mu) = 1$. Therefore, $\theta$ becomes an extension of degree $s$ of $\sigma$. To construct $\field[q^{\mu}]$--linear skew codes it is probably convenient to start with a given automorphism $\sigma \in \Aut{\field[q^\mu]}$ of order $\mu$. We may write it as $\sigma = \tau^h$ with $(h,\mu) = 1$, and $\tau$ is the Frobenius $\field[q]$--automorphism on \(\field[q^\mu]\). For any $s \geq 1$, let $k = a\mu + h$, where $a$ is the product of those prime numbers that appear in a complete factorization of $s$, but are not divisors of $h$. A straightforward argument shows that $\theta = \tau^k \in \Aut{\field[q^{\mu s}]}$ is an extension of degree $s$ of $\sigma$, where \(\tau\) also denotes the Frobenius \(\field[q]\)--automorphism on \(\field[q^{\mu s}]\).  Observe that \(s\) is arbitrary, so the length \(n = \mu s\) of the skew block code defined as a left ideal of $\field[q^\mu][x;\sigma]/\langle x^n -1\rangle$ does not need to be coprime with the characteristic of \(\field[q]\).
\end{example}

\begin{example}[\emph{Skew cyclic convolutional codes}]\label{ex:sccc}
Let $F = \field{}$ be a finite field, and $\pi \in \Aut[E]{F}$ with $|\pi| = s$ and $E = F^\pi$, and extend canonically $\pi$ to the rational function field $F(z)$ in one variable $z$.  A straightforward argument shows that $F(z)^\pi = E(z)$. Now, let $\sigma$ be an $E$--automorphism of $E(z)$ with $|\sigma| = \mu$, defined by $\sigma(z) = (az + b)/(cz + d)$, for some $a, b, c, d \in E$ such that $ad-bc \neq 0$. By L\"{u}roth's Theorem, see \cite[\S 10.2]{vanderWaerden:1970}, there exists \(u \in E(z)\) such that \(E(z)^\sigma = E(u)\). Observe that \([E(z):E(u)] = \mu\). We use the same symbol to denote the extension of \(\sigma\) to \(F(z)\). The order of \(\sigma\) is \(\mu\).

We have \(\sigma\) and \(\pi\) commute because \(a,b,c,d \in E = F^\pi\). Let us call
\[
\theta = \sigma \pi = \pi \sigma \in \Aut[E]{F(z)}.
\]
Since \(\pi\) and \(\sigma\) commute it follows that \(|\theta| = [\mu,s]\). Moreover, \(E(u) \subseteq F(z)^\theta\). We claim that, if  \((\mu,s) = 1\), then \(E(u) = F(z)^\theta\). Indeed,  \(|\theta| = [\mu,s] = s\mu\), so \([F(z):F(z)^\theta] = s\mu\). Since
\[
[F(z):E(u)] = [F(z):E(z)][E(z):E(u)] = s\mu,
\]
it follows that \([F(z)^\theta:E(u)] = 1\). Therefore,  \(\theta \in \Aut[E]{F(z)}\)  is an extension of degree $s$ of \(\sigma \in \Aut[E]{E(z)}\). We thus will be able to construct skew convolutional codes over $\field[]$ of length $n$ with $x^n-1$ fully decomposable in some extension of $\field[](z)[x;\sigma]$ whenever $n = \mu s$ with $\mu$ and $s$ coprime.
\end{example}

\section{Hartmann-Tzeng bound for skew cyclic codes}\label{HT}

In this section, we prove a version for skew cyclic codes of the Hartmann-Tzeng bound \cite{Hartmann/Tzeng:1972}. We keep the notation of Section \ref{algsetup}. 

Let us recall the Circulant Lemma, which is a particular case of \cite[Corollary 4.13]{Lam/Leroy:1988}.
  
\begin{lemma}[Circulant Lemma]\label{circulantlemma}
Let \(\{\alpha_0, \dots, \alpha_{n-1}\}\) be a \(K\)--basis of \(M\). Then, for all \(t \leq n\) and every subset \(\{k_1, k_2, \dots, k_{t}\} \subseteq \{0, 1, \dots, n-1\}\),
\[
\begin{vmatrix}
\alpha_{k_1} & \theta(\alpha_{k_1}) & \dots & \theta^{t-1}(\alpha_{k_1}) \\
\alpha_{k_2} & \theta(\alpha_{k_2}) & \dots & \theta^{t-1}(\alpha_{k_2}) \\
\vdots & \vdots &  & \vdots \\
\alpha_{k_{t}} & \theta(\alpha_{k_{t}}) & \dots & \theta^{t-1}(\alpha_{k_{t}})
\end{vmatrix} \neq 0.
\]
\end{lemma}

\begin{proof}
This is a particular case of \cite[Corollary 4.13]{Lam/Leroy:1988}, whose elementary proof is available in \cite{gln2017sugiyama}.
\end{proof}

Let $\gamma \in M$. According to \cite[Lemma 2.4]{Lam/Leroy:1988}, if \(f = \sum  f_i x^i \in S\), then the right division of \(f\) by \(x-\gamma\) is given by 
\begin{equation}\label{polyevaluation}
f = q  (x-\gamma) + \textstyle\sum f_i \norm{i}{\gamma}, 
\end{equation}
where 
\[
\norm{i}{\gamma} = \gamma \theta(\gamma) \dots \theta^{i-1}(\gamma)
\]
is the \(i\)th norm of \(\gamma\) with respect to \(\theta\). Given \(\alpha \in M\), if \(\beta = \alpha^{-1}\theta(\alpha)\), then
\begin{equation}\label{normproperties}
\norm{i}{\theta^k(\beta)} = \theta^k(\alpha)^{-1} \theta^{k+i}(\alpha).
\end{equation}

Therefore, if \(\alpha \in M\) is such that \(\{\alpha, \theta(\alpha), \dots, \theta^{n-1}(\alpha)\}\) is a normal basis of \(M/K\), then, by \eqref{normproperties} and Lemma \ref{circulantlemma},
\[
\left| \begin{matrix}
1 & \beta & \norm{2}{\beta} & \dots & \norm{n-1}{\beta} \\
1 & \theta(\beta) & \norm{2}{\theta(\beta)} & \dots & \norm{n-1}{\theta(\beta)} \\
\vdots & \vdots & \vdots & \ddots & \vdots \\
1 & \theta^{n-1}(\beta) & \norm{2}{\theta^{n-1}(\beta)} & \dots & \norm{n-1}{\theta^{n-1}(\beta)}
\end{matrix} \right| \neq 0,
\]
As a consequence, by \cite[Lemma 5.7]{Lam/Leroy/Ozturk:2008},
\[
x^n-1 = \lclm{x-\beta, x-\theta(\beta), \dots, x-\theta^{n-1}(\beta)}
\]
(see \cite{gln2016new} for a more direct construction).

In order to prove the Hartmann-Tzeng bound for skew cyclic codes, we need the following technical result. Given a matrix \(A\),  by \(\rank(A)\) we mean its rank. 

\begin{lemma}\label{lemrank}
Let \(\{\alpha_0, \dots, \alpha_{n-1}\}\) be a \(K\)--basis of \(M\), let \(A\) be the \((t+r)\times t\) matrix defined by
 \[
\begin{pmatrix}
\alpha_{k_1} & \theta^{s_1}(\alpha_{k_1}) & \dots & \theta^{(t-1)s_1}(\alpha_{k_1}) \\
\alpha_{k_2} & \theta^{s_1}(\alpha_{k_2}) & \dots & \theta^{(t-1)s_1}(\alpha_{k_2}) \\
\vdots & \vdots & \ddots & \vdots \\
\alpha_{k_{t+r}} & \theta^{s_1}(\alpha_{k_{t+r}}) & \dots & \theta^{(t-1)s_1}(\alpha_{k_{t+r}})
\end{pmatrix},
\]
where \(\{k_1, k_2, \dots, k_{t+r}\} \subseteq \{0, 1, \dots, n-1\}\), and let 
\[
B_i=\left(\begin{array}{c|c|c|c|c}
A & \theta^{s_2}(A) & \theta^{2s_2}(A)& \cdots & \theta^{is_2}(A)
\end{array}\right),
\]
for \(i=0,\ldots, r\), where \((s_1,n)=1\) and \((s_2,n)<t+1\). Then
\[
\rank(B_i) \geq t+i. \]
\end{lemma}

\begin{proof}
Consider the subspaces \(\mathcal B_0,\ldots, \mathcal B_r \subseteq M^{t+r}\), defined as \(\mathcal B_i=\cs(B_i)\) for \(i=0,\ldots, r\), where \(\cs(L)\) denotes the subspace spanned by the columns of \(L\). They form a chain of subspaces
\[
\mathcal B_0 \subseteq \mathcal B_1\subseteq \ldots \subseteq \mathcal B_r.
\]
Since \((s_1,n) = 1\), we have that \(K = M^\theta = M^{\theta^{s_1}}\), so Lemma \ref{circulantlemma} can be applied and \(\dim \mathcal B_0=\rank(A)=t\).  
Observe that, by definition, \( \mathcal B_{i+1}=\mathcal B_i + \theta^{s_2}(\mathcal B_i) \), therefore  \(\mathcal B_i = \mathcal B_{i+1}\) if and only if \(\mathcal B_i = \theta^{s_2}(\mathcal B_i) \).
Suppose by contradiction that  \(\mathcal B_i = \theta^{s_2}(\mathcal B_i) \) and \(\dim(\mathcal B_i) \leq t+i\) for some \(i<r\). Then we also have \(\mathcal B_i = \theta^{(s_2,n)}(\mathcal B_i) \). Since \( \mathcal B_0 \subseteq \mathcal B_i\) and \(\mathcal B_i\) is stable under the action of \(\theta^{(s_2,n)}\), this implies \(\theta^{j(s_2,n)s_1}(\mathcal B_0)\subseteq \mathcal B_i\) for every \(j=1,\ldots,r\). Now, since \((s_2,n)<t+1\) the subspace \(\mathcal B_0+\theta^{(s_2,n)s_1}(\mathcal B_0)\)
contains the column space of the matrix

\[ 
\begin{pmatrix}
\alpha_{k_1} & \theta^{s_1}(\alpha_{k_1}) & \dots & \theta^{(t-1+(s_2,n))s_1}(\alpha_{k_1}) \\
\alpha_{k_2} & \theta^{s_1}(\alpha_{k_2}) & \dots & \theta^{(t-1+(s_2,n))s_1}(\alpha_{k_2}) \\
\vdots & \vdots &  & \vdots \\
\alpha_{k_{t+r}} & \theta^{s_1}(\alpha_{k_{t+r}}) & \dots & \theta^{(t-1+(s_2,n))s_1}(\alpha_{k_{t+r}})
\end{pmatrix},
\]
and, iterating this construction, \(\mathcal B_0+\theta^{(s_2,n)s_1}(\mathcal B_0)+ \ldots + \theta^{r(s_2,n)s_1}(\mathcal B_0) \) contains the column space of the \((t+r)\times (t+r(s_2,n))\) matrix
\[ 
\begin{pmatrix}
\alpha_{k_1} & \theta^{s_1}(\alpha_{k_1}) & \dots & \theta^{(t-1+r(s_2,n))s_1}(\alpha_{k_1}) \\
\alpha_{k_2} & \theta^{s_1}(\alpha_{k_2}) & \dots & \theta^{(t-1+r(s_2,n))s_1}(\alpha_{k_2}) \\
\vdots & \vdots &  & \vdots \\
\alpha_{k_{t+r}} & \theta^{s_1}(\alpha_{k_{t+r}}) & \dots & \theta^{(t-1+r(s_2,n))s_1}(\alpha_{k_{t+r}})
\end{pmatrix},
\]
whose rank is exactly \(t+r\) by Lemma \ref{circulantlemma} applied to \(\theta^{s_1}\). Hence, \(\mathcal B_i\) has dimension \(t+r\) and this yields  a contradiction.
\end{proof}

The distance properties of the skew cyclic code $\mathcal{C}$ defined by $g \in R$ depend on its \emph{\(\beta\)-defining set}, defined as 
\begin{equation} \label{eq:definingset}
T_{\beta}(g) = \{0 \leq i \leq n-1 ~|~ x-\theta^i(\beta) \mid_r g\}.
\end{equation}

\begin{theorem}[Hartmann-Tzeng bound]\label{HTbound}
Assume there exist \(b, \delta, r, t_1, t_2\) with \((n,t_1)=1\) and \( (n,t_2)<\delta \) such that \(\{b+i t_1+\ell t_2 \mid 0\leq i \leq \delta-2,\, 0\leq \ell \leq r\} \subseteq T_{\beta}(g)\). Then \(\distance(\mathcal{C}) \geq \delta+r\). 
\end{theorem}

\begin{proof}
Let \(w = \delta+r-1\) and let \(c \in \mathcal{R}g\) such that \(\weight(c) \leq w\), i.e. \(c = \sum_{j=1}^w c_j x^{k_j}\) for suitable \(\{k_1, \dots, k_w\} \subseteq \{0, \dots, n-1\}\). For each \(0 \leq i \leq \delta - 2\) and \(0 \leq l \leq r\), \(x - \theta^{b+i t_1+\ell t_2}(\beta) \mid_r c\), so by \eqref{polyevaluation} and \eqref{normproperties}
\[
\begin{split}
0 &= \textstyle\sum_{j=1}^w c_j \norm{k_j}{\theta^{b+i t_1+\ell t_2}(\beta)} \\
&= \theta^{b+i t_1+\ell t_2}(\alpha)^{-1} \textstyle\sum_{j=1}^w c_j {\theta^{b+i t_1+\ell t_2 + k_j}(\alpha)}.
\end{split} 
\]
We get that \(c\) is in the left kernel of the matrix \(\theta^{b}(B)\) where
\[
B=\left(\begin{array}{c|c|c|c}
A & \theta^{t_2}(A) & \cdots & \theta^{rt_2}(A)
\end{array}\right)
\]
and 
\[
A=\Big( \theta^{it_1+k_j}(\alpha) \Big)_{\genfrac{}{}{0pt}{}{1 \leq j \leq w}{0 \leq i \leq \delta-2}}.
\]
By Lemma \ref{lemrank} we get that \(\rank(B)=w\) and hence \(c=0 \) is the only element in \(\mathcal R g\) of weight at most
\(\delta+r-1\).
\end{proof}

The classical BCH bound for cyclic codes can also be derived for skew codes as a particular case of the previous theorem. 

\begin{corollary}[BCH bound]\label{BCHbound}
Assume there exist \(b,\delta,t\) with $(t,n)=1$ such that \(\{b, b+t, b+2t, \dots, b+(\delta-2)t\} \subseteq T_{\beta}(g)\). Then \(\distance(\mathcal{C}) \geq \delta\). 
\end{corollary}

\begin{proof}
It is Theorem \ref{HTbound} with \(r = 0\), \(t_1 = t\) and \(t_2 = 1\). 
\end{proof}

\begin{remark}\label{HTUrem}
Recently, a HT bound with respect to a rank metric has been obtained in \cite[Corollary 5]{MartinezPenas:2017}, in the realm of skew block codes. Setting $L=\mathbb{F}_{q^{m}}$ in Theorem \ref{HTbound} we get a similar statement to \cite[Corollary 5]{MartinezPenas:2017}.  The precise relation between both results is explained in Proposition \ref{HTU} and Remark \ref{HTUdif} in the Appendix. We also discuss how Corollary \ref{BCHbound}, when applied to finite fields, boils down to \cite[Proposition 1]{Chaussade/etal:2009}, see Corollary \ref{BCHChaussade} in the Appendix. 
\end{remark}

\begin{remark}
Although in this paper we have focuses on the HT bound, the reader might ask about a generalization analogous to the one provided by Roos \cite{Roos:1979} for cyclic block codes. Following the techniques developed in this paper, and certain combinatorial effort, we think that such a bound may be proved. Another bounds on the minimum distance to be considered for future works may be found in \cite{Lint/Wilson:1986, Zeh/etal:2013, Boston:20013, Duursma/Pellikaan:2006}.
\end{remark}

\section{Constructing Skew Cyclic Codes with designed distance}\label{codesconstruction}

In this section we are going to provide a method for constructing skew cyclic codes with a designed minimum distance. To this end, we investigate the structure of the $\beta$--defining set $T_\beta$ of a given skew polynomial. We will need some general facts on the contraction of left ideals of $S = M[x;\theta]$ to $R = L[x;\sigma]$. We keep the notation of Section \ref{algsetup}. 

Let \(\varrho \in \Galois(M/K)\). We say that a left ideal \(I\) of \(S\) is \emph{\(\varrho\)-invariant} if \(\varrho(I)=I\). Recall that \(\pi = \theta^\mu\) is a generator of the cyclic group \(\Galois(M/L)\), and \(R = S^\pi\).

\begin{lemma}\label{piinvariant}
A left ideal \(I\) is \(\pi\)-invariant if and only if \(I=Sf\) with \(f\in R\).
\end{lemma}

\begin{proof}
If \(I=Sf\) with \(f\in R\) then, for any \(gf\in I\), \((gf)^\pi = g^\pi f^\pi = g^\pi f \in I\), so \(\pi(I)\subseteq I\). Moreover, if \(h = gf \in I\), then \(h = (g^{\pi^{-1}}f)^\pi\), so \(I\subseteq \pi(I)\).

Let us now suppose that \(I\) is \(\pi\)-invariant and \(I = Sf\) with \(f\) monic. Hence $f^\pi=gf$ for some \(g\). But \(\deg f=\deg f^\pi\), so \(g=1\) and \(f=f^\pi\).
\end{proof}

Given a polynomial \(f\in S\), there is a polynomial \(\pseudobound[]{f}\), uniquely determined up to left multiplication by nonzero constants in $L$, such that $Sf \cap R = R\pseudobound[]{f}$. Recall that $\lclm{-}$ denotes the least common left multiple in $S$. 

\begin{proposition}\label{pseudocalc}
If \(f\in S\), then $S\pseudobound{f}$ is the largest $\pi$--invariant left ideal of $S$ contained in $Sf$ and 
\(
S\pseudobound{f} = \bigcap_{i=0}^{s-1}Sf^{\pi^i}.
\)
Consequently, \(\pseudobound{f}=\lclm{f,f^\pi,\dots, f^{\pi^{s-1}}}\).
\end{proposition}

\begin{proof}
 Let $I$ be any $\pi$--invariant left ideal of $S$ contained in $Sf$. By Lemma \ref{piinvariant}, $I = Sg$ for some $g \in R$. Therefore, $g \in Sg \cap R \subseteq Sf \cap R = S\pseudobound{f}$, which implies $Sg \subseteq S\pseudobound{f}$. Since $S\pseudobound{f}$ is $\pi$--invariant (Lemma \ref{piinvariant} again), we get that it is the largest one contained in $Sf$.
 
Observe that, for any \(g\in S\), \(\pi(Sg)= Sg^\pi\). Hence, for any \(0 \leq i \leq s-1\), 
\[S\pseudobound{f} = S(\overline{f})^{\pi^i}=\pi^i(S\overline{f}) \subseteq \pi^i(Sf)=Sf^{\pi^i}.
\] 
Then 
\(
S\pseudobound{f} \subseteq \bigcap_{i=0}^{s-1} Sf^{\pi^i}.
\)
Now, 
\[
\textstyle \pi\left(\bigcap_{i=0}^{s-1} Sf^{\pi^i}\right) = \bigcap_{i=0}^{s-1}\pi(Sf^{\pi^i}) = \bigcap_{i=0}^{s-1} Sf^{\pi^i},
\] 
since \(\pi\) has order \(s\), so it is a \(\pi\)--invariant left ideal contained in \(Sf\), whence \(\bigcap_{i=0}^{s-1} Sf^{\pi^i} \subseteq S\pseudobound{f}\) and the equality holds. Finally, \(\bigcap_{i=0}^{s-1} Sf^{\pi^i} = S \lclm{f,f^\pi,\dots, f^{\pi^{s-1}}}\), hence the consequence also follows.
\end{proof}

We consider the set $C_n = \{0, 1, \dots, n-1 \}$ as a cyclic group of order $n$. Since $n = s \mu$, the subset $C_s = \{ 0, \mu, \dots, (s-1)\mu\}$ is the subgroup of order $s$ of $C_n$. By $C_n/C_s$ we denote the quotient group. Recall that given a nonzero $g \in R$, and $\beta$ as in Section \ref{HT}, we can consider the $\beta$--defining set of $g$  
\[
T: = T_\beta(g) = \{ i \in C_n : (x-\theta^i(\beta))\mid_r g \}  
\]
We start with the following observation.

\begin{lemma}
If $g \in R$ is such that $T$ is not empty, then $T = T^1 \cup \cdots \cup T^k$, for some $T^1, \dots, T^k \in C_n/C_s$.
\end{lemma}
\begin{proof}
For each $i \in T$, we get from Lemma \ref{piinvariant} and the inclusion $Sg \subseteq S(x - \theta^i(\beta))$ that
\[
Sg = (Sg)^\pi \subseteq S(x - \theta^i(\beta))^\pi = S(x - \theta^{i+\mu}(\beta)),
\]
since \(\pi = \theta^\mu\). Therefore, $i + \mu \in T$ for all $i \in T$, which proves the Lemma. 
\end{proof}

Given any nonempty subset $T \subseteq C_n$, we may consider the polynomial $g_T = \lclm{x-\theta^i(\beta)}^{i\in T } \in S$.  On the other hand, we denote by $\overline{T} \subseteq C_n$ the smallest union of cosets in $C_n/C_s$ such that $T \subseteq \overline{T}$.  

\begin{lemma}\label{TTbarra}
For a nonempty $T \subseteq C_n$, $g_{\overline{T}} = \overline{g_T}$. 
\end{lemma}
\begin{proof}
First, observe that, for any $f, g \in S$, we have $\overline{\lclm{f,h}} = \lclm{\overline{f},\overline{h}}$, since $Sf \cap Sg \cap R =  R\overline{f} \cap R\overline{h}$. Therefore, by using Proposition \ref{pseudocalc} in the second equality of the following computation, we get
\[
\begin{split}
\overline{\lclm{x - \theta^i(\beta)}^{i \in T}} &= \lclm{\overline{x - \theta^i(\beta)}}^{i \in T} = \lclm{\lclm{x-\theta^{i+j\mu}}^{0 \leq j \leq s-1}}^{i \in T} \\ 
&= \lclm{\lclm{x-\theta^k(\beta)}^{k \in [i]}}^{i \in T} = \lclm{x - \theta^k(\beta)}^{k \in \overline{T}}
\end{split}
\]
where \([i]\) is the equivalence class of \(i\) in \(C_n/C_s\).
\end{proof}

Now we are ready to construct skew cyclic codes of designed distance. Let $b, \delta, r, t_1, t_2$ be non negative integers such that $(n,t_1) = 1, (n,t_2) < \delta$ and $\delta + r \leq n-1$. Under these conditions, we define 
\[
T_{b, \delta, r, t_1, t_2} = \{b+i t_1+\ell t_2 \mid 0\leq i \leq \delta-2,\, 0\leq \ell \leq r\}.
\]

\begin{theorem}\label{Hdist}
Define 
\[
g = g_{\overline{T_{b,\delta,r,t_1,t_2}}} = \lclm{x-\theta^{i}(\beta)}^{i \in \overline{T_{b,\delta,r,t_1,t_2}}} \in R.
\]
If $\overline{T_{b,\delta,r,t_1,t_2}} \neq C_n$, then the skew cyclic code \(\mathcal{C} = \tovector(\mathcal{R}g) \) over $L$ has minimum distance $\distance(\mathcal C)\geq \delta +r$. The dimension of $\mathcal{C}$ is $(\mu - t)s$, where $t$ is the number of cosets in $C_n/C_s$ needed to build $\overline{T_{b,\delta,r,t_1,t_2}}$. 
\end{theorem}

\begin{proof}
Observe that, for $g$ as in the statement, $T_\beta =\overline{T_{b,\delta,r,t_1,t_2}} \supseteq T_{b,\delta,r,t_1,t_2}$.  Thus, the inequality  $\distance(\mathcal C)\geq \delta +r$ is a direct consequence of Theorem \ref{HTbound}. On the other hand, $\deg g$ is the cardinality of $\overline{T_{b,\delta,r,t_1,t_2}}$, that is, $ts$. Thus, the dimension of $\mathcal{C}$ is $n - ts = (\mu - t)s$, as desired. 
\end{proof}

\begin{example}
Let $L=\field[2](b)$ be the field with $2^{5}$ elements, where $b^5 + b^2 + 1 = 0$, and consider the automorphism $\sigma:L\to L$ given by $\sigma(b)=b^2$. It is easy to check that the order of $\sigma$ is 5. Let also $L\subseteq M$ be the field extension with $M = \field[2](a)$ the field with $2^{10}$ elements, where $a^{10} + a^6 + a^5 + a^3 + a^2 + a + 1=0$. The embedding $\epsilon:L\to M$ is defined by $\epsilon(b)=a^8 + a^6 + a^2 = a^{33}$. For brevity, except for 0 and 1, we write the elements of $M$ as powers of $a$. Let now $\theta : M\to M$ be the Frobenius automorphism, so its order is 10. Observe that $\theta \epsilon (b)=\theta(a^{33})=a^{66}=a^9 + a^8 + a^3 + a^2 + a
 =(a^8 + a^6 + a^2)^2=\epsilon(b)^2=\epsilon(b^2)=\epsilon \sigma(b)$, so the restriction of $\theta$ to $L$ is $\sigma$.

Let us fix $\alpha=a^5$, which provides a normal basis of $M$ as an $M^\theta$-vector space. Hence $\beta=a^5$. Now, let us choose $\delta =4$ and $r=1$, the parameters of the Hartmann-Tzeng bound, and  $b=0$, $t_1=3$ and $t_2=2$. Therefore, $T_{0,4,2,3,2} =\{0, 2, 3, 5, 6, 8\}$, and $\overline{T_{0,4,2,3,2}} = \{ 0, 1, 2, 3, 5, 6, 7, 8 \}.$ Thus, by Theorem \ref{Hdist}, the polynomial $g = \lclm{x-\theta^i(\beta)}^{i = 0,1,2,3,5,6,7,8}$ belongs to $R = L[x;\sigma]$ and defines a skew cyclic block code of dimension $2$ and distance at least $5$. Explicitly, 
\[
\begin{split}
g &= y^{8} + \left(b^{3} + b^{2}\right) y^{7} + \left(b^{4} + b^{3} +
b\right) y^{6} + \left(b^{4} + b^{2} + b + 1\right) y^{5} \\ 
&\quad + \left(b^{4}+ b^{3}\right) y^{4} + \left(b^{2} + b\right) y^{3} + \left(b^{2} +
b\right) y^{2}  + \left(b^{4} + b^{3} + b^{2}\right) y + b^{2} + b.
\end{split}
\]
\end{example}

We finish this section including two tables. Table \ref{tab:scbc} provides a list of skew cyclic block codes. As observed in Example \ref{ex:scbc} there is no restriction on the values of \(\mu\) and \(s\), but it should exist an automorphism of \(L = \field[q]\) of order \(\mu\). In each case, \(M = \field[q^s]\). So large lengths require large fields. The legend is as follows: \(L = \field[p](b)\), where \(p = 2,3\), \(M = \field[p](a)\), and \(SB\) denotes the Singleton bound. The generator of each code, computed by Theorem \ref{Hdist} as a least left common multiple, is not written for brevity.

\begin{table}[ht]
\caption{Some skew cyclic block codes}\label{tab:scbc}
\begin{turn}{90}
\begin{tabular}{|ccccccccccccc|}
\hline
$L$ & $\sigma(b)$ &$\mu$ & $ s $ & $n$ &  $\theta(a)$ & $\alpha$ & $k$ & $SB$ & $\delta+r$ & $(t_1,t_2)$ & $T$ & $\overline{T}$ \\  \hline \hline
$\field[8]$ & $b^{ 2 }$ & 3 & 4 & 12 & $a^{ 2
}$ & $a^ 5 $ & 4 & 9 & 3+0 & (5,1) & [0,5]
& [0,2,3,5,6,8,9,11]  \\ \hline
$\field[16]$ & $b^{ 2 }$ & 4 & 2 & 8 & $a^{ 2
}$ & $a^ 5 $ & 2 &7 & 3+1 & (1,3) & [0,1,3,4] & [0,1,3,4,5,7]  \\ \hline
$\field[32]$ & $b^{ 2 }$ & 5 & 2 & 10 & $a^{
2 }$ & $a^{10} $ & 2 &9 &  3+1 & (3,1) & [0,1,3,4] & [0,1,3--6,8,9]  \\ \hline
$\field[64]$ & $b^{ 2 }$ & 6 & 4 & 24 & $a^{
2 }$ & $a^ 9 $ & 8 &17&   4+1 & (1,7) & [0--2,7--9] &
 [0--3,6--9,12--15,18--21] \\ \hline
$\field[64]$ & $b^{ 4 }$ & 3 & 3 & 9 & $a^{ 4
}$ & $a^ 5 $ & 3  & 7&  2+1 & (2,2) & [0,2]
& [0,2,3,5,6,8] \\ \hline
$\field[128]$ & $b^{ 4 }$ & 7 & 2 & 14 & $a^{
512 }$ & $a^{14} $ & 2 & 13 & 3+3 & (3,2) & [0,2--7,9] & 
[0,2--7,9--13]
 \\ \hline
$\field[256]$ & $b^{ 2 }$ & 8 & 2 & 16 & $a^{
2 }$ & $a^{13} $ & 2 & 15 & 4+4 & (1,7) & 
[0--2,5--9,12--15] 
& 
[0--2,4--10,12--15]  \\ \hline
$\field[256]$ & $b^{ 2 }$ & 8 & 2 & 16 & $a^{
2 }$ & $a^{13} $ & 2 &15 & 2+6 & (1,3) & [0,2,3,6,9,12,15] & 
[0--4,6--12,14,15]
 \\ \hline
$\field[256]$ & $b^{ 2 }$ & 8 & 2 & 16 & $a^{
2 }$ & $a^ 5 $ & 8 & 9& 3+1 & (1,3) & [0,1,3,4] & [0,1,3,4,8,9,11,12]  \\ \hline
$\field[1024]$ & $b^{ 2 }$ & 10 & 2 & 20 &
$a^{ 2 }$ & $a^{11} $ & 2 & 19 & 5+5 & (3, 7) &
[0,1,3,4,6--11,13--18] 
 &  [0,1,3--11,13--19]  \\ \hline
$\field[27]$ & $b^{ 3 }$ & 3 & 4 & 12 & $a^{3 }$
 & $a^ 7 $ & 4 & 9 & 2+1 & (5,1) & [0,1]
& [0,1,3,4,6,7,9,10] \\ \hline
$\field[81]$ & $b^{ 3 }$ & 4 & 4 & 16 & $a^{
3 }$ & $a^{10} $ & 4 & 13 & 4+0 & (3,0) & [0,3,6] & [0,2--4,6--8,10--12,14,15] \\ \hline
$\field[125]$ & $b^{ 5 }$ & 3 & 3 & 9 & $a^{
5 }$ & $a^ 8 $ & 3 & 7& 2+1 & (2,5) & [0,5]
& [0,2,3,5,6,8]  \\ \hline
\end{tabular}
\end{turn}
\end{table}

Table \ref{tab:sccc} contains some skew cyclic convolutional codes. As observed in Example \ref{ex:sccc}, \(\mu\) and \(s\) are coprime. The structure of \(\Aut[E]{E(z)}\) is much richer than the group of automorphisms of a finite field, so we can get big lengths with relatively small base finite fields. Once again the generators are not written for brevity.

\begin{table}[ht]
\caption{Some skew cyclic convolutional codes}\label{tab:sccc}
\begin{turn}{90}
\begin{tabular}{|ccccccccccccccc|}
\hline
$E$ & $F$ & $\pi(a)$ & $\epsilon(b)$ & $\sigma(z)$ & $\mu$ & $s$  & $n$ & $\alpha$ & $k$ &  $SB$  & $\delta +r$ & $(t_1,t_2)$ & $T$ & $\overline{T}$ \\  
\hline  \hline
$\field[2]$ & $\field[16]$ & $a^{ 2 }$ & $a^{ 0 }$ & $\frac{1}{z + 1}$ & 3 & 4 & 12 & $a^{3} t$ & 4 & 9 & 2 +1 & (1,5) & [0,5] & [0,2,3,5,6,8,9,11] \\ \hline
$\field[4]$ & $\field[16]$ & $a^{ 4 }$ & $a^{ 5 }$ & $\frac{1}{z + b}$ & 5 & 2 & 10 & $\frac{a}{t}$ & 2 & 9 & 3 + 2 & (1,4) & [0,1,4,5,8,9] & [0,1,3,4,5,6,8,9]  \\ \hline
$\field[4]$ & $\field[64]$ & $a^{ 4 }$ & $a^{ 21 }$ & $\frac{1}{z + b^2}$ & 5 & 3 & 15 & $a^{5} t$ & 3 & 13 & 4 + 1 & (1,6) & [0,1,2,6,7,8] & [0--3,5--8,10--13]
\\ \hline
$\field[16]$ & $\field[256]$ & $a^{ 16 }$ & $a^{ 17 }$ & $\frac{b^{2}}{b^8z + b^{12}} $ & 15 & 2 & 30 & $a t$ & 2 & 29 & 2 + 13 & (1,7) & 
$\begin{array}{c} [0,1,3,5,7,10,\\ 12,14,17,19, \\21,24,26,28]\end{array}$ &
 [0--7,9--22,24--29] \\ \hline
$\field[32]$ & $\field[1024]$ & $a^{ 32 }$ & $a^{ 33 }$ & {\scriptsize $b^{4} z + b^{30}$ }& 31 & 2 & 62 & $\frac{1}{a t}$ & 40 & 23 & 12 + 0 & (1,5) & [0--10] & [0--10, 31--41] \\ \hline
$\field[3]$ & $\field[9]$ & $a^{ 3 }$ & $a^{ 0 }$ & $\frac{2 z}{z + 2}$ & 3 & 2 & 6 & $a t^{2}$ & 2 & 5 &  2 + 1 & (1,5) & [0,5] & [0,2,3,5] \\ \hline
$\field[9]$ & $\field[729]$ & $a^{ 9 }$ & $a^{ 91 }$ & $\frac{b z + b + 2}{\left(b + 2\right) z + b}$ & 4 & 3 & 12 & $a t$ & 3 & 10 &  3 + 1 & (1,5) & [0,1,5,6] & [0--2,4--6,8--10] \\
\hline
\end{tabular}
\end{turn}
\end{table}

\begin{example} Let us consider the parameters of the first item in Table \ref{tab:sccc}. Therefore, we are dealing with the base field $\mathbb{F}_2(z)$, and the automorphism $\sigma$ defined by $\sigma(z)=1/(1+z)$, whose order is $\mu=3$. We also consider the extension of degree 4 $(\mathbb{F}_{16}(t),\theta)$, where $\theta(a)=a^2$ and $\theta(t)=t/(1+t)$.

Let us fix $\alpha=a^3t$, which provides a normal basis of $\mathbb{F}_{16}(t)$ as an $(\mathbb{F}_{16}(t))^\theta$-vector space. Therefore $\beta=(a + 1)/(at^2 + at)
$. Now, we choose $\delta =2$ and $r=1$, the parameters of the Hartmann-Tzeng bound, and  $b=0$, $t_1=1$ and $t_2=5$. Therefore, $T_{0,2,1,1,5} =\{0, 5\}$, and $\overline{T_{0,2,1,1,5}} = \{ 0, 2, 3, 5, 6, 8, 9, 11 \}.$ Thus, by Theorem \ref{Hdist}, the polynomial $g = \lclm{x-\theta^i(\beta)}^{i = 0, 2, 3, 5, 6, 8, 9, 11}$ defines a skew $\sigma$-cyclic convolutional code over $\mathbb{F}_2(z)$ of length 12, dimension $4$ and distance at least $3$. Explicitly, 
\[
g = x^8+\frac{1}{z^3+z}x^4+\frac{z^2}{z^3+z^2+z+1},
\]
viewed in $\mathbb{F}_2(z)[x;\sigma]/\langle x^n-1\rangle$.

\end{example}

\section{Decoding skew BCH codes}\label{sec:decoding}

Let \(\delta \leq n-1\) and \(t \leq n-1\) be such that \((t,n) = 1\). We denote \(T_{\delta,t} = T_{0,\delta,0,t,1} = \{0, t, \dots, (\delta-2)t\}\), and \(T_\delta = T_{\delta,1} = \{0, 1, \dots, \delta-2\}\). 

In the finite fields case the notion of skew BCH codes was introduced in \cite{Boucher/etal:2007} first for $q=2$, and then generalized in \cite{Chaussade/etal:2009} and \cite{Boucher/Ulmer:2014} for every $q$. The following definition extends it to the general framework of finite extension fields with cyclic Galois group.

\begin{definition}
Let \(f = g_{\overline{T_{\delta,t}}}\). The skew cyclic code \(\mathcal{C} = \tovector(\mathcal{R}f) \subseteq L^n\) is called a \emph{skew BCH code of designed distance \(\delta\)}. As a direct consequence of Theorem \ref{Hdist} or Corollary \ref{BCHbound}, \(\distance(\mathcal{C}) \geq \delta\).
\end{definition}

The milestone of this section is to provide a decoding algorithms for skew BCH codes, and it is achieved in two steps. First we are going to embed our code in an \(M\)-linear MDS code \(\mathcal{D}\) in such a way that the decoding algorithm with respect to \(\mathcal{D}\) provides a decoding algorithm with respect to \(\mathcal{C}\). Then a permutation equivalent code is computed. This new code fits in the family of codes which can be decoded using one of the nearest neighbor decoding algorithms appearing in \cite{GLNPGZ,gln2017sugiyama,Boucher/etal:2007}. Recall that a nearest decoding algorithm finds the closest codeword to the received vector, see e.g. \cite[\S 1.11.2]{Huffman/Pless:2010}.

We fix \(\delta\) and \(t\) with \((t,n) = 1\). Let \(f = g_{\overline{T_{\delta,t}}} \in R\) and \(g = g_{T_{\delta,t}} \in S\). Let \(\mathcal{C} = \tovector(\mathcal{R} f) \subseteq L^n\) and \(\mathcal{D} = \tovector(\mathcal{S}g) \subseteq M^n\). Since \(g \mid_r f\) it follows that \(\mathcal{C} \subseteq \mathcal{D}\), in fact \(\mathcal{C}\) is a subfield subcode of \(\tovector(\mathcal{S}f)\), which is a subcode of \(\mathcal{D}\). We know \(\delta \leq \distance(\mathcal{C})\) by Theorem \ref{Hdist}. The distance of \(\mathcal{D}\) can also be computed. 

\begin{lemma}\label{Ddistance}
The code \(\mathcal{D}\) is an \(M\)-linear code of length \(n\), dimension \(n-\delta+1\) and distance \(\delta\), hence an MDS code. 
\end{lemma}

\begin{proof}
We apply Theorem \ref{Hdist} with \(s = 1\), i.e. \(M = L\) and \(\theta = \sigma\). Since, in this situation, \(\overline{T_{\delta,t}} = T_{\delta,t}\), we have that the dimension of \(\mathcal{D}\) is \(n-(\delta-1)\) and \(\distance(\mathcal{D}) \geq \delta\). The equality of the distance follows from the Singleton bound. 
\end{proof}

\begin{proposition}\label{dectoMDS}
Assume a codeword \(c \in \mathcal{C}\) is transmitted and \(v = c + e \in L^n\) is received with \(\weight(e) \leq \left\lfloor \frac{\delta - 1}{2} \right\rfloor\). Any nearest neighbor decoding algorithm on \(\mathcal{D}\) correctly computes \(c\) and \(e\). 
\end{proposition}

\begin{proof}
It follows from Lemma \ref{Ddistance} that \(c\) is the only vector in \(\mathcal{D}\) such that \(\distance(v,c) \leq \left\lfloor \frac{\delta - 1}{2} \right\rfloor\), hence any nearest neighbor decoding algorithm on \(\mathcal{D}\) must compute \(c\) and \(e\). 
\end{proof}

\begin{remark}\label{decodeskewRS}
Skew Reed-Solomon codes (skew RS codes, for short) are introduced in \cite{gln2017sugiyama} (in the convolutional case) and \cite{GLNPGZ} (for general fields), inspired by the classical block Reed-Solomon codes. They are skew cyclic codes generated by \(T_\delta\) when \(s=1\). There are fast nearest neighbor decoding algorithms for skew RS codes. Such algorithms can be found in \cite{Boucher/etal:2007,GLNPGZ} for block codes and \cite{gln2017sugiyama,GLNPGZ} for convolutional codes. In the non commutative setting over finite fields, the term Generalized skew Reed-Solomon (GSRS) has been used in \cite[Definition 9]{Liu/etal:2015}. These codes are presented as evaluation codes, and skew RS block codes are of this kind (see Remark \ref{RSGSRS}).  Always over finite fields, this kind of evaluation codes are introduced in \cite[Definition 7]{Boucher/Ulmer:2014}. In the latter reference, skew RS block codes are proved to be MDS codes in \cite[Theorem 5]{Boucher/Ulmer:2014}, so that Lemma \ref{Ddistance} could be seen as an extension of this result to skew RS convolutional codes, and even to more general situations. In the convolutional case the possible connection between non commutative evaluation codes and skew RS codes would require a deeper understanding of the dual of skew RS codes.
\end{remark}

The code \(\mathcal{D}\) does not fit in the definition of skew RS code since \(t\) is not necessarily equal to \(1\), hence the decoding algorithms in Remark \ref{decodeskewRS} do not apply to \(\mathcal{D}\). Our next step consists in building a new skew RS code which is permutation equivalent to \(\mathcal{D}\). Then the decoding with respect to this new code can be translated to \(\mathcal{D}\), and hence to \(\mathcal{C}\). 

When \(M = \field[q^n]\) and \(K = \field[q]\), this last step can be avoided. By \eqref{polyevaluation} and \eqref{normproperties}, see also Lemma \ref{rootroot}, a parity check matrix of \(\mathcal{D}\) is the transpose of 
\[
\left( \begin{matrix}
\alpha & \theta^t(\alpha) & \cdots & \theta^{t(\delta-2)}(\alpha) \\
\theta(\alpha) & \theta(\theta^t(\alpha)) & \cdots & \theta(\theta^{t(\delta-2)}(\alpha)) \\
\vdots & \vdots & \ddots & \vdots \\
\theta^{n-1}(\alpha) & \theta^{n-1}(\theta^t(\alpha)) & \cdots & \theta^{n-1}(\theta^{t(\delta-2)}(\alpha))
\end{matrix} \right).
\]
Hence \(\mathcal{D}\) is a generalized Gabidulin code as defined in \cite[\S IV.A]{Kshevetskiy/Gabidulin:2005} taking \(h_i = \theta^{i-1}(\alpha)\). Therefore, the decoding algorithm proposed in \cite[\S VII]{Kshevetskiy/Gabidulin:2005} can be directly applied. Skew evaluation codes, as defined in \cite{Liu/etal:2015}, and Gabidulin codes in arbitrary characteristic, as introduced in \cite{Augot/etal:2013} are also studied when \(t=1\), hence in order to apply the Skew Berlekamp-Welch Algorithm \cite[Algorithm 1]{Liu/etal:2015} for finite fields, or the Unique Decoding \cite[\S V.C]{Augot/etal:2013} for fields of arbitrary characteristic, 
our proposed next step have to be followed. 

Since \((t,n) = 1\), there exists \(u \in C_n\) such that \(ut = 1 + vn\). Let \(\rho : C_n \to C_n\) defined by \(\rho(i) = it\). This map is a bijection whose inverse is \(\rho^{-1}(j) = ju\). We also denote by \(\rho : M^n \to M^n\) the linear map defined by \(\rho(a_0, \dots, a_{n-1}) = (a_{\rho(0)}, \dots, a_{\rho(n-1)})\). Let \(\{\epsilon_0, \dots, \epsilon_{n-1}\}\) be the canonical basis of \(M^n\). Then \(\rho(\epsilon_j) = \epsilon_h\), where \(\rho(h) = j\), i.e.
\begin{equation}\label{rhobasis}
\rho(\epsilon_j) = \epsilon_{\rho^{-1}(j)} = \epsilon_{ju}.
\end{equation}
Since \(\rho : M^n \to M^n\) is a permutation of the positions, \(\mathcal{D}\) and \(\rho(\mathcal{D})\) are permutation equivalent codes, hence they share the same parameters. 

\begin{remark}\label{decPEcodes}
Assume a codeword \(c \in \mathcal{D}\) is transmitted and \(v = c + e \in M^n\) is received with \(\weight(e) \leq \left\lfloor \frac{\delta - 1}{2} \right\rfloor\). Since \(\mathcal{D}\) and \(\rho(\mathcal{D})\) are permutation equivalent, \(\rho(c)\) is the only vector in \(\rho(\mathcal{D})\) such that \(\distance(\rho(v),\rho(c)) \leq \left\lfloor \frac{\delta - 1}{2} \right\rfloor\). Any nearest neighbor decoding algorithm on \(\rho(\mathcal{D})\) applied to \(\rho(v)\) correctly computes \(\rho(c)\) and \(\rho(e)\).  
\end{remark}

We want to prove that there is a skew RS structure on \(\rho(\mathcal{D})\). In order to do so, we need to recognize in \(\rho\) the arithmetical structure on \(M^n\) provided by the automorphism \(\tovector: \mathcal{S} \to M^n\). Recall that \(\mathcal{S} = S / S(x^n-1)\), where \(S = M[x;\theta]\). We adopt the notation \(S' = M[y;\theta^t]\) and \(\mathcal{S}' = S'/S'(y^n-1)\). Since \((t,n) = 1\), \(|\theta^t| = n\), \(M^{\theta^t} = M^\theta = K\) and \(\{\alpha, \theta^t(\alpha), \dots, \theta^{(n-1)t}(\alpha)\}\) is a normal basis of \(M/K\). 

\begin{lemma}\label{algebramap}
The map \(\varphi : \mathcal{S} \to \mathcal{S}'\) defined by \(\varphi(a x^i) = a y^{iu}\), is an isomorphism of $K$--algebras such that \(\rho \tovector = \tovector \varphi\). 
\end{lemma}

\begin{proof}
Let \(\tilde{\varphi} : S \to S'\) be the map defined by \(\tilde{\varphi}(ax^i) = a y^{iu}\). Since 
\[
\tilde{\varphi}(xa) = \tilde{\varphi}(\theta(a)x) = \theta(a) y^u = \theta^{ut-vn}(a)y^u = (\theta^t)^u(a) y^u = y^u a = \tilde{\varphi}(x) \tilde{\varphi}(a), 
\]
it follows that \(\tilde{\varphi}\) is a morphism of algebras. Since 
\(
\tilde{\varphi}(x^n-1) = y^{un} - 1 \in S'(y^n-1),
\) 
we conclude that \(\varphi\) is a well defined morphism of algebras. It is clear that \(\varphi\) is \(M\)-linear and, hence, $K$--linear. In order to check that \(\varphi\) is isomorphism, it is sufficient to see that it is surjective. Since, for all \(i \in C_n,\)
\[
\varphi(x^{it}) = y^{itu} = y^{i(1 + vn)} = y^i y^{ivn} = y^i,
\]
it follows that \(\varphi\) is surjective. Finally
\[
\tovector (\varphi(x^j)) = \tovector(y^{ju}) = \epsilon_{ju} = \rho(\epsilon_j) = \rho(\tovector(x^j))
\]
by \eqref{rhobasis}, so \(\rho \tovector = \tovector \varphi\). 
\end{proof}

The following proposition is a direct consequence of Lemma \ref{algebramap}.

\begin{proposition}\label{rhoDisSCC}
With notation as above, \(\rho(\mathcal{D}) = \tovector(\mathcal{S}'\varphi(g))\). 
\end{proposition}

The generator \(\varphi(g)\) is not necessarily a right divisor of \(y^n-1\), but \(\mathcal{S}'\varphi(g) = \mathcal{S}'g'\), where \(g' = \gcrd{\varphi(g),y^n-1}\), hence \(\rho(\mathcal{D})\) is a skew cyclic code.

\begin{proposition}\label{dectoskewRS}
Let \(\beta' = \alpha^{-1} \theta^t(\alpha)\). Let \(g = g_{T_{\delta,t}} = \lclm{x-\theta^j(\beta)}^{j \in T_{\delta,t}} \in \mathcal{S}\) and \(g' = \gcrd{\varphi(g),y^n-1} \in \mathcal{S}'\). Then \(g' = g_{T_{\delta}} = \lclm{y-\theta^{it}(\beta')}^{i \in T_\delta}\). In particular \(\rho(\mathcal{D}) = \tovector(\mathcal{S}' g')\) is a skew RS code. 
\end{proposition}

\begin{proof}
Let \(\norm{l}{\gamma}\) denote the \(l\)th-norm with respect to \(\theta^t\), i.e. 
\[
\norm{l}{\gamma} = \gamma \theta^t(\gamma) \dots \theta^{tl}(\gamma).
\]
Then \(\norm{l}{\beta'} = \alpha^{-1} (\theta^t)^l(\alpha)\). On one hand, \(\norm{n}{\beta'} = 1\), which implies that 
\[
y-(\theta^t)^i(\beta') \mid_r y^n-1
\] 
for all \(i \in T_\delta\). On the other hand,
\[
\norm{u}{(\theta^t)^i(\beta')} = \theta^{it}(\norm{u}{\beta'}) = \theta^{it}(\alpha^{-1} \theta^{tu}(\alpha)) = \theta^{it}(\alpha^{-1} \theta(\alpha)) = \theta^{it}(\beta),
\]
which proves that 
\[
y - (\theta^t)^i(\beta') \mid_r y^u - \theta^{it}(\beta) = \varphi(x - \theta^{it}(\beta)) \mid_r \varphi(g). 
\]
Hence 
\[
g_{T_\delta} = \lclm{y-\theta^{it}(\beta')}^{i \in T_\delta} \mid_r g' = \gcrd{\varphi(g),y^n-1}.
\]
We have that \(\deg g = \delta - 1\) and therefore \(\dim_M \mathcal{D} = n - \delta + 1\). So \(\dim_M \rho(\mathcal{D}) = n - \delta + 1\). Then \(\deg g' = \delta - 1\) because \(g'\) is a minimal generator. Since \(\deg g_{T_\delta} = \delta-1\), we conclude that \(g' = g_{T_\delta}\).  
\end{proof}

Algorithm \ref{decoding} describes a decoding algorithm for skew BCH codes. 

\begin{algorithm}[H]
\caption{Decoding algorithm for skew BCH codes}\label{decoding}
\begin{algorithmic}[1]
\REQUIRE \(T_{\delta,t}\), \(f = g_{\overline{T_{\delta,t}}}\), and a code \(\mathcal{C} = \tovector(\mathcal{R}f)\). A received transmission \(v = \left (v_0,\ldots ,v_{n-1}\right ) \in L^n\) with no more than \(\left\lfloor \frac{\delta-1}{2} \right\rfloor\) errors with respect to a codeword in \(\mathcal{C}\). 
\ENSURE The error \(e = \left (e_0,\ldots ,e_{n-1}\right )\) such that \(v-e \in \mathcal{C}\)
\STATE Let \(\rho(v) = (v_{\rho(0)}, \dots, v_{\rho(n-1)}) \in M^n\).
\STATE Let \(g' = g_{T_\delta} \in M[y;\theta^t]\)
\STATE Let \((x_0, \dots, x_{n-1})\) be the error computed when a nearest decoding algorithm is applied to \(\rho(v)\) with respect to the code \(\tovector(\mathcal{S}'g')\).
\RETURN \(\rho^{-1}(x_0,\ldots ,x_{n-1})\) 
\end{algorithmic}
\end{algorithm}

\begin{theorem}
Let \(f = g_{\overline{T_{\delta,t}}}\) and \(\mathcal{C} = \tovector(\mathcal{R}f)\). Algorithm \ref{decoding} correctly computes the error vector of any received vector if the weight of the error is not greater that \(\left\lfloor \frac{\delta-1}{2} \right\rfloor\).
\end{theorem}

\begin{proof}
By Proposition \ref{dectoMDS} and Remark \ref{decPEcodes}, any nearest neighbor decoding algorithm applied to \(\rho(v)\) correctly computes \(\rho(e)\) if \(v = c + e\) with \(c \in \mathcal{C}\) and \(\weight(e) \leq \left\lfloor \frac{\delta-1}{2} \right\rfloor\). By Remark \ref{decodeskewRS} and Propositions \ref{rhoDisSCC} and \ref{dectoskewRS}, fast nearest neighbor decoding algorithms exist to decode \(\rho(v)\). 
\end{proof}

\begin{remark}
In general, the complexity of Algorithm \ref{decoding} depends on the base field, and the algorithm used for decoding skew RS codes. See, for instance, the discussion in \cite[Remark 21]{gln2017sugiyama}. As an example, if the base field is finite and we use the Peterson-Gorenstein-Zierler algorithm described in \cite{GLNPGZ}, the complexity becomes cubic with respect to operations on the field. See Example \ref{exampff}.
\end{remark}

\begin{example}\label{exampff}
Let $L=\mathbb{F}_2(b)$ be the field with $2^{8}$ elements, where $b^8 + b^4 + b^3 + b^2 + 1=0$, and $\sigma:L\to L$ the automorphism defined by $\sigma(b)=b^8$, i.e. $\sigma=\tau^3$, where $\tau$ is the Frobenius automorphism. Therefore, the order of $\sigma$ is $\mu=8$. Let us denote $R=L[x;\sigma]$ and $\mathcal{R}=L[x;\sigma]/\langle x^{16}-1\rangle$. Let us consider the field extension $L\subset M$, where $M= \mathbb{F}_2(a)$ is the field with $2^{16}$ elements such that $a^{16} + a^5 + a^3 + a^2 + 1=0$, and $\theta=\tau^3$. Therefore, $s=[M:L]=2$ and $|\theta|=16$. Observe that the embedding $\epsilon:L\to M$ is given by $\epsilon(b)=a^{77}$. Set $S=M[x;\theta]$ and $\mathcal{S} = S/S(x^{16}-1)$. Then \(\epsilon\) extends canonically to \(\epsilon : \mathcal{R} \to \mathcal{S}\). By the method provided in Section \ref{codesconstruction}, we choose $b=0$, $\delta=7$ and $t=11$, so that $T=\{0, 11, 6, 1, 12, 7\}$. Actually, $\overline{T}=\{0, 1, 3, 4, 6,7 ,8 , 9, 11, 12, 14, 15\}$. Set $\alpha=a^{11}$, that provides a normal basis of $M$ over $M^\theta$, and $\beta=\alpha^{-1}\theta(\alpha)=a^{77}$. Therefore,
\[
g=x^6+a^{60395}x^5+a^{25401}x^4+a^{31814}x^3+a^{58173}x^2+a^{15228}x+a^{15937}
\]
defines a skew cyclic code $\mathcal{D}=\mathfrak{v}(\mathcal{S}g)$ over $M$ of length 16 and dimension 10, and 
\begin{multline*}
\overline{g}=x^{12}+b^{48}x^{11}+b^{146}x^{10}+b^{158}x^9+b^{29}x^8+b^{17}x^7+b^{52}x^6+\\ b^{127}x^5+b^{169}x^4+b^{208}x^3+b^{229}x^2+b^{102}x+b^{115}
\end{multline*}
defines a skew BCH code $\mathcal{C}=\mathfrak{v}(\mathcal{R}\overline{g})$ over $L$ of length 16 and dimension 4. Moreover, by Corollary \ref{BCHbound}, the Hamming distance of $\mathcal{C}$ is greater or equal than 7 and, following Algorithm \ref{decoding}, it corrects up to 3 errors. Then $\mathcal{D}$ is an MDS code of Hamming distance 7. 

Let now \(S' = M[y;\theta^t]\), and the ring isomorphism
\[
\varphi: \mathcal{S} \longrightarrow \mathcal{S}' = \frac{S'}{S' (y^n-1)}
\]
defined by $\varphi(x^i)=y^{3i}$, for any $i$, since $3$ is the inverse of $11$ modulo $16$. Hence 
\begin{multline*}
g' = \varphi(g) = \lclm{y-\rho^i(\gamma)}^{i=0,1,2,3,4,5} = y^6 + a^{45739}y^5 + a^{60997}y^4 \\ 
+ a^{10959}y^3 + a^{3299}y^2 + a^{14798}y + a^{41129}
\end{multline*}
generates a skew RS code $\rho(\mathcal{D}) = \tovector(\mathcal{S}' g')$.

Suppose then we need to send the message $m=b^{56}x^3+bx^2+b^{13}x+b^{34}$ so the encoded polynomial to be transmitted is 
\begin{multline*}
c = m\overline{g} =b^{56}x^{15}+b^{179}x^{14}+b^{93}x^{13}+b^{28}x^{12}+b^{31}x^{11}+b^{53}x^{10}+b^{209}x^9+\\ b^{93}x^8+b^{178}x^7+ b^{78}x^6+b^{249}x^5+b^{50}x^4+b^{79}x^3+b^{198}x^2+b^{171}x+b^{149}
\end{multline*}
After the transmission, we receive a polynomial
\begin{multline*}
v =b^{56}x^{15}+b^{179}x^{14}+b^{20}x^{13}+b^{28}x^{12}+b^{31}x^{11}+b^{53}x^{10}+b^{76}x^9+\\ b^{93}x^8+b^{178}x^7+ b^{78}x^6+b^{175}x^5+b^{50}x^4+b^{79}x^3+b^{198}x^2+b^{171}x+b^{149},
\end{multline*}
that is $v=c+e$, where $e=bx^{13}+b^{71}x^9+b^{23}x^5$.

Now, we use the code $\rho(\mathcal{D})$ in order to correct the received polynomial $v$, viewed in $\mathcal{S}'$ via $\varphi\epsilon$. Concretely,
\[
\begin{split}
\varphi\epsilon(v) &= a^{24415}y^{15} + a^{27242}y^{14} + a^{28784}y^{13} + a^{25700}y^{12} \\
&\quad + a^{39064}y^{11} + a^{26471}y^{10} + a^{40606}y^9 + a^{47802}y^8 + a^{10280}y^7 + a^{36237}y^6 \\
&\quad + a^{25957}y^5 + a^{14392}y^4 + a^{22359}y^3 + a^{40092}y^2 + a^{15934}y + a^{11051}.
\end{split}
\]
Let us apply \cite[Algorithm 1]{GLNPGZ}. We first calculate the full matrix of syndromes,
$$\left(\begin{array}{rrr}
a^{48031} & a^{33571} & a^{16897} \\
a^{1607} & a^{33794} & a^{53272} \\
a^{2053} & a^{41009} & a^{40611} \\
a^{16483} & a^{15687} & a^{13015}\end{array}\right)$$
and its reduced column echelon form
$$\left(\begin{array}{rrr}
1 & 0 & 0 \\
0 & 1 & 0 \\
0 & 0 & 1 \\
a^{2516} & a^{38350} & a^{16308}
\end{array}\right).$$
Therefore, the rank of the matrix is three, and the monic polynomial in the left kernel
is the error locator polynomial $\lambda=y^3+a^{16308}y^2+a^{38350}y+a^{2516}$. So that we found three errors at positions 7, 11 and 15. Finally, we solve the system
$$\left(\begin{array}{rrr}
a^{1408} & a^{22528} & a^{32773} \\
a^{2816} & a^{45056} & a^{11} \\
a^{5632} & a^{24577} & a^{22}
\end{array}\right)\cdot 
\left(\begin{array}{c}
e_7\\
e_{11}\\
e_{15}
\end{array} \right )=
\left(\begin{array}{c}
a^{48031}\\ 
a^{1607}\\
a^{2053}
\end{array}\right ),$$
whose solution is $e_7=a^{514}$, $e_{11}=a^{36494}$ and $e_{13}=a^{11822}$. Hence, the error polynomial is $\varphi(e)=a^{11822}y^{15}+a^{36494}y^11+a^{514}y^7$. Therefore, $e = a^{514}x^{13}+a^{36494}x^9+a^{11822}x^5$, or viewed in $R$, $e = bx^{13}+b^{71}x^9+b^{23}x^3$, as expected.
\end{example}

\begin{example} Under the conditions and notation of Example 4.7, consider the skew cyclic convolutional code $\mathcal{C}$ generated by
\[
f = x^8+\frac{1}{z^3+z}x^4+\frac{z^2}{z^3+z^2+z+1},
\]
 that belongs to $\mathbb{F}_2(z)[x;\sigma]/\langle x^{12}-1\rangle$. Let us suppose that we need to send the message provided by 
 \[m=x^{3} + z^{2} x^{2} + \frac{1}{z} x + z.\]
We encode the message and get the codeword 
\[\begin{split} mf=  & x^{11} + z^{2} x^{10} + \frac{1}{z} x^{9} + z x^{8} +
\left(\frac{1}{z^{3} + z}\right) x^{7} + \left(\frac{z^{5}}{z +
1}\right) x^{6} 
\\ +  & \left(\frac{z^{3} + z^{2} + z + 1}{z^{3}}\right) x^{5}
 +  \left(\frac{1}{z^{2} + 1}\right) x^{4} + \left(\frac{z^{2}}{z^{3} +
z^{2} + z + 1}\right) x^{3} 
\\ + & \left(z^{5} + z^{3}\right) x^{2} +
\left(\frac{z + 1}{z^{4}}\right) x + \frac{z^{3}}{z^{3} + z^{2} + z + 1}.\end{split}\]
Suppose now that, after the transmission, the received polynomial $y=mf+e$, where the error polynomial $e=zx^9$. That is, there is an error in the 9th position, which we expect to recover.

We then translate the problem to the working algebra $\mathbb{F}_{16}(t)[x;\theta]/\langle x^{12}-1\rangle$  via the embedding $\epsilon:(\mathbb{F}_2(z),\sigma)\to (\mathbb{F}_{16}(t),\theta)$. Consider the skew cyclic code $\mathcal{D}\subset \mathbb{F}_{16}(t)^{12}$ generated by \[g=[x-\beta,x-\theta^5(\beta)]_{\ell}.\] In particular, $g$ verifies $\overline{g}=\epsilon(f)$, where we also denote $\epsilon$ the appropriated extension to polynomials. Therefore, $\epsilon(mf)$ belongs to $\mathcal{D}$. Actually, the error $\epsilon(e)=tx^9$ remains to be of weight 1. Then, the problem consists in correcting the polynomial $\epsilon(y)$ by using the skew cyclic code $\mathcal{D}$. In order to see $\mathcal{D}$ as a skew RS code, we need the isomorphism 
\[\varphi: \mathcal{S}=\frac{\mathbb{F}_{16}(t)[x;\theta]}{\langle x^{12}-1\rangle} \to \frac{\mathbb{F}_{16}(t)[y;\rho]}{\langle y^{12}-1\rangle}=\mathcal{S}'\] defined by $\varphi(x)=y^5$, where $\rho=\theta^5$. Hence, $\varphi(\mathfrak{v}^{-1}(\mathcal{D}))=\varphi{(\mathcal{S}g)}=\mathcal{S}'\varphi{(g)}=\mathcal{S}'g'$ is a skew RS code generated by 
\[\begin{split} g'= & [y-\gamma,y-\rho(\gamma)]_{\ell} \\ = &  y^{2} + \left(\frac{t^{4} + \left(a^{2} + 1\right) t^{3} + \left(a^{2} + a + 1\right) t^{2} + \left(a + 1\right) t}{\left(a^{3} + 1\right) t^{2}
+ \left(a + 1\right) t + a^{3} + a^{2} + 1}\right) y \\ + & \frac{t^{4} +
\left(a^{2} + 1\right) t^{3} + \left(a^{2} + a\right) t^{2}}{\left(a^{3}
+ a^{2} + a\right) t^{3} + \left(a^{3} + a^{2}\right) t^{2} +
\left(a^{2} + 1\right) t + a^{2} + a + 1}.\end{split}\]
This code has error-correcting capacity 1, so it can correct the received polynomial
\[\begin{split} \varphi\epsilon(y)= & \left(\frac{1}{t^{3} + t}\right) w^{11} + \left(t^{5} + t^{3}\right)
w^{10} + \left(\frac{t^{2} + 1}{t}\right) w^{9} + \left(\frac{1}{t^{2} +
1}\right) w^{8} + w^{7} \\ +  & \left(\frac{t^{5}}{t + 1}\right) w^{6} +
\left(\frac{t + 1}{t^{4}}\right) w^{5} + t w^{4} +
\left(\frac{t^{2}}{t^{3} + t^{2} + t + 1}\right) w^{3} \\+ & t^{2} w^{2} +
\left(\frac{t^{3} + t^{2} + t + 1}{t^{3}}\right) w + \frac{t^{3}}{t^{3}
+ t^{2} + t + 1}. \end{split}\]
In order to do this, we may make use of the Sugiyama Algorithm proposed in \cite{gln2017sugiyama}. Indeed, the syndrome polynomial becomes $$S=\left(\left(a^{3} + a^{2} + a + 1\right) t + a^{3} + a^{2} + a +
1\right) w + \left(a^{3} + a^{2}\right) t^{2}.$$
Hence, by applying the Right Extended Euclidean Algorithm to $S$ and $y^2$, it yields that the error locator polynomial $\lambda=y + (a^{3} + a^{2} + a + 1)/(t^{2} + t)=y-\rho^9(\gamma)$ and the error evaluator polynomial $\omega=a^{3} t/(t + 1)$. From these considerations, we may deduce that there is an error at position 9 whose value is $t$. Hence, the error polynomial is $e'=ty^9$, and thus $e=\epsilon^{-1}\varphi^{-1}(e')=zx^9$, as expected.
\end{example}

\appendix

\section{Hartmann-Tzeng and BCH bounds for skew block codes through linearized polynomials, and a remark on Gabidulin non block codes.}

The theory so far developed may be applied, in particular, to finite fields, that is, for skew block codes. Thus, some of our statements boil down to existing results in the literature.  Concretely, Theorem \ref{HTbound} and Corollary \ref{BCHbound} might be mathematically considered as extensions of \cite[Corollary 5]{MartinezPenas:2017} and \cite[Proposition 1]{Chaussade/etal:2009} to abstract Galois field extensions with cyclic Galois group. Since the latter heavily rest upon the good properties of the spaces of roots of linearized of (commutative) polynomials (or difference equations), in order to clarify such a connection, we need to discuss how non-commutative and commutative roots are related.  The key point is Lemma \ref{rootroot} below, kindly pointed out by one of the referees. 

\subsection{Hartmann-Tzeng bound.}
Recently, in \cite[Corollary 5]{MartinezPenas:2017}, a version of the  Hartmann-Tzeng bound with respect to a rank metric has been proved for $q^k$--skew block codes. Since the distance given by the rank metric is lower than the Hamming distance, \cite[Corollary 5]{MartinezPenas:2017} implies a Hartmann-Tzeng bound for the latter. It is pertinent to compare Theorem \ref{HTbound}, when applied to finite fields, with \cite[Corollary 5]{MartinezPenas:2017}. In order to work with exactly the same general hypotheses than in \cite{MartinezPenas:2017}, we note that Section \ref{HT} is still valid, word by word, without requiring that $L^{\sigma} = M^{\theta}$ as in Definition \ref{autoextension}. To be precise, let $M$ be a finite field extension of $L$, and $\theta : M \to M$ a field automorphism of finite order $n$ such that, by restriction, gives an automorphism $\sigma : L \to L$. Set $R = L[x;\sigma]$ and $S = M[x;\theta]$. Obviously, $R$ is a subring of $S$, and $x^n -1$ is a central element both in $R$ and in $S$. 

\begin{lemma}\label{2ideals}
\(R(x^n-1) = R \cap S(x^n-1)\).
\end{lemma}

\begin{proof}
The inclusion \(R(x^n-1) \subseteq R \cap S(x^n-1)\) is trivial. So assume \(\sum_{i=0}^\nu f_i x^i \in S\) such that \(\sum_{i=0}^\nu f_i x^i(x^n-1) \in R\). We have
\[
\sum_{i=0}^\nu f_i x^i(x^n-1) = \sum_{i=0}^\nu f_i x^{i+n} - \sum_{i=0}^\nu f_i x^i = \sum_{i=0}^{n-1} -f_i x^i + \sum_{i=n}^{n+\nu} (f_{i-n} - f_i) x^i \in R.
\]
Then for all \(0 \leq i \leq n-1\), \(-f_i \in L\). For all \(n \leq i \leq 2n-1\), \(f_{i-n} - f_i \in L\), hence \(-f_i \in L\). Iterating the process, we get that \(f_i \in L\) for all \(0 \leq i \leq \nu\), i.e. \(\sum_{i=0}^\nu f_i x^i \in R\). Therefore \(R(x^n-1) \supseteq R \cap S(x^n-1)\).
\end{proof}

In view of Lemma \ref{2ideals}, we still have a ring extension $\mathcal{R} \subseteq \mathcal{S}$, where
\[
\mathcal{R} = \frac{R}{R(x^n-1)} \quad \text{and} \quad \mathcal{S} = \frac{S}{S(x^n-1)}. 
\]
 
If we set $K = M^{\theta}$, then all statements and proofs of Section \ref{HT} are still valid. 


Now, let us briefly recall the scenario where \cite{MartinezPenas:2017} is represented. Let $k,m,n$ positive integers such that $m$ divides $kn$, so that $\field[q^m]$ is considered as a subfield of $\field[q^{kn}]$. Let $t$ be a commutative variable, and let us denote by $\Lq{kn}$ the ring of $q^k$--linearized polynomials over $\field[q^{kn}]$. Obviously, $\Lq{m}$ is a subring of $\Lq{kn}$. 

Denote by $\tau$ the Frobenius $\field[q]$--automorphism of $\field[q^{kn}]$. The order of the automorphism $\theta = \tau^k$ of $\field[q^{kn}]$ is $n$. On the other hand, $\theta$ restricts to an automorphism $\sigma : \field[q^m] \to \field[q^m]$ whose order, say $\mu$, obviously divides $n$. Let $M= \field[q^{kn}]$, $L = \field[q^m]$ and $K = M^\theta = \field[q^k]$. 

\begin{remark}
Although the results from Section \ref{HT} do not require the assumption $K = L^\sigma$, if we want to apply the whole theory developed in the paper to the finite field case under the conditions of \cite{MartinezPenas:2017}, then we should assume it. Since $L^\sigma = \field[q^{(m,k)}]$, it follows that the assumption $K = L^\sigma$ is just to require $k$ to be a divisor of $m$. 
\end{remark}

The connection between skew cyclic block codes in the sense of \cite{MartinezPenas:2017} and in the skew polynomial setting is based in the ring isomorphism (see, e.g., \cite[Theorem II.13]{McDonald:1974})
\begin{equation}\label{Phi}
\Phi : R = L[x;\sigma] \to \Lq{m}, \qquad \left( \textstyle\sum_{i}f_i x^i \mapsto \sum_i f_i t^{[ik]} \right), 
\end{equation}
where, as usual, $[j] = q^{j}$ for every non negative integer $j$. Indeed, the isomorphism \eqref{Phi} projects to an isomorphism
\[
\mathcal{R} = \frac{L[x;\sigma]}{\langle x^n-1 \rangle} \cong \frac{\Lq{m}}{\langle t^{[kn]} - t \rangle},
\]
that gives an explicit correspondence between skew block codes and $q^k$--cyclic codes in view of  \cite{MartinezPenas:2017,Boucher/etal:2007,Gabidulin:2009}. The isomorphism \eqref{Phi} can be seen as the restriction to $R$ of the isomorphism
\[
\Phi : S = M[x;\theta] \to \Lq{kn}, \qquad \left( \textstyle\sum_{i}f_i x^i \mapsto \sum_i f_i t^{[ik]} \right),
\]
which allows to display a concrete relationship between non-commutative roots of a skew polynomial and the usual roots of its linearized polynomial. The following lemma is formulated, for example, in \cite[Lemma 4]{Chaussade/etal:2009}.

\begin{lemma}\label{rootroot}
Let $f \in S$ and $\alpha \in \field[q^{kn}]$. Then $\Phi(f)(\alpha) = 0$ if and only if $\alpha^{-1}\theta(\alpha)$ is a right roof ot $f$. 
\end{lemma}
\begin{proof}
Straightforward from \eqref{polyevaluation} and \eqref{normproperties}. See also \cite[(3)]{GLNPGZ}.
\end{proof}

Following the notation of \cite{MartinezPenas:2017}, the $\field[q^k]$--vector space of the roots in $\field[q^{kn}]$ of $F \in \Lq{m}$ is denoted by $Z(F)$. 

From now on, fix $\alpha \in M = \field[q^{kn}]$ such that $\{ \theta^i(\alpha) : i = 0, \dots, n-1 \}$ is a normal basis of $M = \field[q^{kn}]$ over $K = \field[q^k]$. We write $\beta = \alpha^{-1} \theta(\alpha)$. 

\begin{lemma}\label{betaceros}
For every $g \in R$, the set $\{ \theta^{i}(\alpha) : i \in T_\beta(g) \}$ is either empty or a (linearly independent) subset of $Z(\Phi(g))$. If $g$ is a right divisor of $x^n -1$, then it is a basis if and only if $g =[x-\theta^i(\beta)]^{i \in T_\beta(g)}_\ell$
\end{lemma}
\begin{proof}
Let $i \in T_\beta(g)$, that is, $\theta^i(\beta)$ is a right root of $g$. Since $\theta^i(\beta) = \theta^i(\alpha)^{-1}\theta(\theta^i(\alpha))$, we get from Lemma \ref{rootroot} that $\theta^i(\alpha)$ is a root of $\Phi(g)$. Now, if $g$ is a right divisor of $x^n-1$, then, by \cite[Theorem 1]{Chaussade/etal:2009}, we have that $\dim_{\field[q^k]} (Z(\Phi(g)) = \deg g$. Thus, $\{ \theta^i(\alpha) : i \in T_\beta(g) \}$ is a basis of $Z(\Phi(g))$ if, and only if, the cardinal of $T_\beta(g)$ equals $\deg g$, by virtue of \cite[Lemma 6]{GLNPGZ}.
\end{proof}

The following Proposition is now a consequence of Theorem \ref{HTbound}.

\begin{proposition}(c.f. \cite[Corollary 5]{MartinezPenas:2017})\label{HTU}
Let $g \in R$ be a right divisor of $x^n-1$, consider $\mathcal{C} \subseteq \field[q^m]^n$ the $\field[q^m]$--linear $q^k$--cyclic code with minimal generator $\Phi(g) \in \Lq{m}$.  Assume that there exist numbers \(b, \delta, r, t_1, t_2\) with \((n,t_1)=1\) and \( (n,t_2)<\delta \) such that \(\{\theta^{b+i t_1+\ell t_2}(\alpha) \mid 0\leq i \leq \delta-2,\, 0\leq \ell \leq r\} \subseteq Z(\Phi(g))\). Then \(\distance(\mathcal{C}) \geq \delta+r\). 
\end{proposition}
\begin{proof}
Since the inclusion \(\{\theta^{b+i t_1+\ell t_2}(\alpha) \mid 0\leq i \leq \delta-2,\, 0\leq \ell \leq r\} \subseteq Z(\Phi(g))\)  is equivalent, in view of the first statement of Lemma \ref{betaceros}, to the inclusion \(\{b+i t_1+\ell t_2 \mid 0\leq i \leq \delta-2,\, 0\leq \ell \leq r\} \subseteq T_{\beta}(g)\), the Proposition follows by applying Theorem \ref{HTbound}. 
\end{proof}

\begin{remark}\label{HTUdif}
There are some differences between the statements of Proposition \ref{HTU} and \cite[Corollary 5]{MartinezPenas:2017}. First, the statement of \cite[Corollary 5]{MartinezPenas:2017} concerns the rank distance, which is lower than the Hamming distance. Second, in \cite{MartinezPenas:2017}, $\alpha$ is not required to generate a normal basis, but, in compensation, the set  $\{\theta^{b+i t_1+\ell t_2}(\alpha) \mid 0\leq i \leq \delta-2,\, 0\leq \ell \leq r \}$   is assumed to be linearly independent over $\field[q^k]$. Third, in \cite[Corollary 5]{MartinezPenas:2017}, the condition $\delta + r \leq m$ is assumed. In our approach, this extra hypothesis is not needed.  
\end{remark}

Although in our general setting, the cyclotomic spaces used in \cite{MartinezPenas:2017} are not needed, 
the reader may be curious about some potential relationship between them and our $\beta$--defining sets. Our next aim is to show that a particular choice of $\beta$--defining set leads to a set of vectors in a cyclotomic space. To this end, recall from \cite[Definition 3]{MartinezPenas:2017} that the minimal polynomial of $\gamma \in \field[q^{kn}]$ is the monic polynomial $F \in \Lq{m}$ of minimal degree such that $F(\gamma) = 0$.  The cyclotomic space of $\beta$ is then defined (\cite[Definition 4]{MartinezPenas:2017}) as the set of all roots in $\field[q^{kn}]$ of $F$.  Given $0 \neq f \in S$ let us denote, as in Section \ref{codesconstruction}, by $\overline{f}$ the unique monic polynomial in $R$ such that $R\overline{f} = Sf \cap R$.

\begin{proposition}\label{cyclo}
Let $\beta = \gamma^{-1}\theta(\gamma)$, for some $\gamma \in \field[q^{kn}]$. Then the minimal polynomial of $\gamma$ is $\Phi(  \overline{x - \beta})$.  As a consequence, if $T$ is the $\beta$--defining set of $\overline{x- \beta}$,  then $\{\theta^i(\gamma) : i \in T\}$ is a subset of the cyclotomic space of $\gamma$. 
\end{proposition}

\begin{proof}
Let $F \in \Lq{m}$ denote the minimal polynomial of $\gamma$. By Lemma \ref{rootroot}, $\beta$ is a right root of $f = \Phi^{-1}(F) \in R$. Therefore, $Sf \subseteq S(x-\beta)$ and, since $f  \in R$, this implies that $Rf \subseteq Sf \cap R \subseteq S(x-\beta) \cap R = R\overline{(x - \beta)}$. Therefore, the degree of $\overline{x - \beta}$ is less than or equal to the degree of $f$. By Lemma \ref{rootroot}, $\Phi(\overline{x - \beta})(\gamma) =0$, and, since the degree of $\Phi(\overline{x - \beta})$ is less than or equal to the degree of $\Phi(f) = F$, we conclude that $\Phi(\overline{x - \beta}) = F$. 

Now, let $i \in T$, that is, $\theta^i(\beta)$ is a right root of $\overline{x - \beta}$. We get, as in the proof of Lemma \ref{betaceros}, that $\theta^i(\gamma)$ is a root of $\Phi(\overline{x - \beta}) = F$. But, according to \cite[Definition 4]{MartinezPenas:2017}, this means that $\{ \theta^i(\gamma) : i \in T \}$ is a subset of the cyclotomic space of $\gamma$. 
\end{proof}

\begin{example}
Let us assume in this example that $K = L^\sigma$ (equivalently, that $k$ divides $m$). Let $\alpha \in \field[q^{kn}]$ such that $\{\alpha, \theta(\alpha), \dots, \theta^{n-1}(\alpha) \}$ is a normal basis of $M = \field[q^{kn}]$ over $K = \field[q^k]$, and write $kn = sm$, and $m = k\mu$. Set $\beta = \alpha^{-1}\theta(\alpha)$. Given $b \geq 0$, Lemma \ref{TTbarra} tells us that, if we write $T = \{b \}$, then 
\[
\overline{x-\theta^b(\beta)} = [x-\theta^b(\beta), x-\theta^{b +\mu}(\beta), \cdots, x-\theta^{b + (s-1)\mu}(\beta)]_{\ell}.
\]
By Proposition \ref{cyclo}, we get that $\{\theta^b(\alpha), \theta^{b +\mu}(\alpha), \cdots, \theta^{b + (s-1)\mu}(\alpha) \}$ is a subset of the cyclotomic space of $\theta^b(\alpha)$ which, in view of Lemma \ref{betaceros}, is already a basis. Setting $k = 1$, we get \cite[Proposition 1]{MartinezPenas:2017}. 
\end{example}

\subsection{BCH bound.} A BCH type bound was proved for skew cyclic block codes in \cite[Proposition 1]{Chaussade/etal:2009}. Since this bound works for the rank distance, it gives in particular a BCH bound for the Hamming distance. Our next aim is to compare \cite[Proposition 1]{Chaussade/etal:2009} with the particular case of Corollary \ref{BCHbound} when applied to finite fields. We should first note that \cite[Proposition 1]{Chaussade/etal:2009} works when the word ambient ring is modeled with  more general central polynomials than $x^n -1$. Obviously, to reach our purpose, we only consider this last case.  The hypotheses of \cite[Proposition 1]{Chaussade/etal:2009} are formulated in terms of the solution space of the difference equation associated to a skew polynomial. As it is already observed in \cite{Chaussade/etal:2009}, we may equivalently work with roots of suitable linearized polynomials.

Let $\sigma$ denote an automorphism of order $\mu$ of a finite field $L = \field[q^\mu]$, so that $K = \field[q^\mu]^\sigma = \field[q]$. Given $n = \mu s$, we may extend $\sigma$, according to Example \ref{ex:scbc}, to an automorphism $\theta$ of $M = \field[q^n]$ such that $M^\theta = \field[q]$. Let $k$ be a positive integer such that $\sigma = \tau^k$, where $\tau$ is the $\field[q]$--automorphism of Frobenius of $\field[q^n]$. Therefore, $\theta(\gamma) = \gamma^{q^k}$ for every $\gamma \in \field[q^n]$. As in \eqref{Phi}, we have a ring isomorphism $\Phi : S = M[x;\theta] \cong \Lq{n}$ that assigns to every skew polynomial $g \in S$ its linearized $q^k$--polynomial $\Phi(g)$. By restriction, we get an isomorphism $\Phi : R = L[x;\sigma] \cong \Lq{\mu}$. Clearly, the roots of $\Phi(g)$, for $g \in R$,  in a given field extension of $\field[q^\mu]$ are, precisely, the solutions  of the difference equation, in the sense of \cite{Chaussade/etal:2009}, associated to $g$. A  consequence of  Corollary \ref{BCHbound}, when applied to finite fields, is the following particular case of \cite[Proposition 1]{Chaussade/etal:2009}.

\begin{corollary}\cite[Proposition 1]{Chaussade/etal:2009}\label{BCHChaussade}
Let $\{ \alpha, \theta(\alpha), \dots, \theta^{n-1}(\alpha)\}$ be a normal basis of $\field[q^n]$ over $\field[q]$, and let $g \in R = \field[q^\mu][x;\sigma]$ be a right divisor of $x^n -1$. Assume there exist $b, \delta$ such that $\{ \theta^{b}(\alpha), \theta^{b +t}(\alpha), \dots, \theta^{b + (\delta -2)t}(\alpha) \}$ are zeroes of $\Phi(g)$. Then the Hamming distance of the skew block code $\mathcal{C} \subseteq \field[q^\mu]^n$ generated by $g$ is at least $\delta$. 
\end{corollary}
\begin{proof}
Set $\beta = \alpha^{-1}\theta(\alpha)$. By Lemma \ref{rootroot}, we have that $\{ b, b +t, \dots, b + (\delta -2)t \} \subseteq T_{\beta}(g)$. Thus, the corollary is obtained from Corollary \ref{BCHbound}. 
\end{proof}

\subsection{Gabidulin generalized non block codes.}

The idea of using root spaces as a tool to define and study linear codes has been exported in \cite{Augot/etal:2013} from the finite field case, previously described in this appendix, to an abstract field extension $K \subseteq L$ of finite degree $\mu$.  To be precise, let $\sigma$ be a $K$--automorphism of $L$, and $R  = L[x;\sigma]$. To each $g = \sum_ig_ix^i \in R$, we can associate (see \cite{Augot/etal:2013})  a $K$--linear endomorphism $g(\theta) = \sum_ig_i\theta^i$ of $L$ whose kernel is called the root space of $g$. Under the hypothesis $K = L^\sigma$, generalized Gabidulin codes of length (and, hence, distance) at most $\mu$ are defined by means to a suitable generating matrix, and they are proved to be of maximum rank distance. The roots spaces of non commutative polynomials are used to design a decoding procedure. In order to describe a connection between these Gabidulin codes and the skew cyclic codes, we take advantage of the description of the parity check matrix of a generalized Gabidulin code given in \cite{Augot/etal:2013}. The key is the following general version of Lemma \ref{rootroot}.

\begin{lemma}\label{rootroot2}
Let $g \in R$ and $\alpha \in L$. Then $\alpha^{-1}\theta(\alpha)$ is a right root of $g$ if and only if $\alpha$ belongs to the root space of $g$. 
\end{lemma}
\begin{proof}
Straightforward from \eqref{polyevaluation} and \eqref{normproperties}. See also \cite[(3)]{GLNPGZ}.
\end{proof}

Even thought that a generalized Gabidulin code need not to be a skew cyclic code, there is a particular case where both constructions meet. Let $\mathcal{C}$ be the generalized Gabidulin code obtained when, in the parity check matrix given in \cite[Definition 3]{Augot/etal:2013},
$h_i = \sigma^{i-1}(\alpha)$ for $i = 1, \dots, \mu$, that is, $\{ \alpha, \sigma(\alpha), \dots, \sigma^{\mu-1}(\alpha) \}$ is a normal basis of $L$ over $K$. This is a code of dimension $k$ and minimum distance $\delta = \mu - k +1$. The definition of a skew RS code is recalled in Remark \ref{decodeskewRS}. 

\begin{proposition}\label{GabRS}
The code $\mathcal{C}$ is, precisely, the skew RS code generated in $\mathcal{R} = R/\langle x^\mu - 1 \rangle$ by $g = [x-\beta, \dots, x-\sigma^{\delta -2}(\beta)]_{\ell}$.
\end{proposition}
\begin{proof}
Identify $(f_0, \dots, f_{\mu  - 1}) \in L^\mu$ with the polynomial $f = \sum_i f_i x^i$ through the isomorphism $L^\mu \cong \mathcal{R}$. In view of the form of the parity check matrix of $\mathcal{C}$, we get that $f \in \mathcal{C}$ if and only if $\alpha, \sigma(\alpha), \dots, \sigma^{\delta-2}(\alpha)$ belong to the root space of $g$. Equivalently, by Lemma \ref{rootroot2}, $f$ belongs to $\mathcal{C}$ if and only if $\beta, \sigma(\beta), \dots, \sigma^{\delta - 2}(\beta)$ are right roots of $f$. Therefore, $\mathcal{C} = \mathcal{R} g$, where $g = [x-\beta, \dots, x-\sigma^{\delta -2}(\beta)]_{\ell}$.
\end{proof}

\begin{remark}\label{RSGSRS}
Skew RS block codes in the sense of \cite{GLNPGZ} are, by virtue of \cite[Theorem 5]{Liu/etal:2015}, and Proposition \ref{GabRS},  particular cases of the GSRS codes introduced in \cite[Definition 9]{Liu/etal:2015}. As a consequence, the Skew Berlekamp-Welch Algorithm from \cite{Liu/etal:2015} is of application to skew  RS block codes. Also, the decoding algorithm from \cite{Augot/etal:2013} can be applied in for skew RS codes over more general fields, once an effective path from generator to parity check matrices as presented in \cite[Definition 3]{Augot/etal:2013} is provided. Concretely, skew RS codes are defined in terms of right roots, i.e. by a parity check matrix which corresponds to the analogous matrix \(H\) in \cite[Definition 3]{Augot/etal:2013}. In order to apply the Unique Decoding algorithm presented in \cite[\S V.C]{Augot/etal:2013}, it is needed to compute \(g_1, \dots, g_n \in M\)  such that 
\[
\left( \begin{matrix}
g_1 & \cdots & g_n \\
\theta(g_1) & \cdots & \theta(g_n) \\
\vdots & \ddots & \vdots \\
\theta^{k-1}(g_1) & \cdots & \theta^{k-1}(g_n) 
\end{matrix} \right)
\]
is a generator matrix of our skew RS code. These elements can be computed by using \cite[Proposition 8]{Augot/etal:2017}. A better understanding of the dual of a skew RS code would also derive this computation.
\end{remark}

\bibliographystyle{abbrv}


\begin{thebibliography}{10}

\bibitem{Augot/etal:2013}
D.~Augot, P.~Loidreau, and G.~Robert.
\newblock Rank metric and {Gabidulin} codes in characteristic zero.
\newblock In {\em {ISIT 2013 IEEE International Symposium on Information
  Theory, Jul 2013, Istanbul, Turkey.}}, 2013.

\bibitem{Augot/etal:2017}
D.~Augot, P.~Loidreau, and G.~Robert.
\newblock {Generalized Gabidulin codes over fields of any characteristic}.
\newblock arXiv:1703.09125v1 [cs.IT], 2017.

\bibitem{Bose/Chaudhuri:1960}
R.~C. Bose and D.~K. Ray-Chaudhuri.
\newblock On a class of error correcting binary group codes.
\newblock {\em Information and Control}, 3(1):68 -- 79, 1960.

\bibitem{Boston:20013}
N.~Boston.
\newblock Bounding minimum distances of cyclic codes using algebraic geometry.
\newblock {\em Electronic Notes in Discrete Mathematics}, 6:385 -- 394, 2001.
\newblock WCC2001, International Workshop on Coding and Cryptography.

\bibitem{Boucher/etal:2007}
D.~Boucher, W.~Geiselmann, and F.~Ulmer.
\newblock Skew-cyclic codes.
\newblock {\em Applicable Algebra in Engeneering, Communication and Computing},
  18:379--389, 2007.

\bibitem{Boucher/Ulmer:2009}
D.~Boucher and F.~Ulmer.
\newblock Coding with skew polynomial rings.
\newblock {\em Journalm of Symbolic Computation}, 44(12):1644--1656, 2009.

\bibitem{Boucher/Ulmer:2014}
D.~Boucher and F.~Ulmer.
\newblock Linear codes using skew polynomials with automorphisms and
  derivations.
\newblock {\em Designs, Codes and Cryptography}, 70(3):405--431, March 2014.

\bibitem{Bueso/alt:2003}
J.~L. Bueso, J.~G\'omez-Torrecillas, and A.~Verschoren.
\newblock {\em Algorithmic methods in non-commutative algebra. Applications to
  quantum groups}, volume~17 of {\em Mathematical Modelling: Theory and
  Applications}.
\newblock Kluwer Academic Publishers, Dordrecht, 2003.

\bibitem{Chaussade/etal:2009}
L.~Chaussade, P.~Loidreau, and F.~Ulmer.
\newblock Skew codes of prescribed distance or rank.
\newblock {\em Designs, Codes and Cryptography}, 50(3):267--284, 2009.

\bibitem{Duursma/Pellikaan:2006}
I.~M. Duursma and R.~Pellikaan.
\newblock A symmetric {Roos} bound for linear codes.
\newblock {\em Journal of Combinatorial Theory, Series A}, 113(8):1677 -- 1688,
  2006.
\newblock Special Issue in Honor of Jacobus H. van Lint.

\bibitem{Gabidulin:1985}
E.~M. Gabidulin.
\newblock Theory of codes with maximum rank distance.
\newblock {\em Probl. Peredachi Inf.}, 21(1):3--16, 1985.

\bibitem{Gabidulin:2009}
E.~M. Gabidulin.
\newblock Rank q-cyclic and pseudo-q-cyclic codes.
\newblock In {\em 2009 IEEE International Symposium on Information Theory,
  Seoul}, pages 2799--2802, 2009.

\bibitem{Gluesing/Schmale:2004}
H.~Gluesing-Luerssen and W.~Schmale.
\newblock On cyclic convolutional codes.
\newblock {\em Acta Applicandae Mathematicae}, 82(2):183--237, 2004.

\bibitem{gln2016new}
J.~G{\'o}mez-Torrecillas, F.~J. Lobillo, and G.~Navarro.
\newblock A new perspective of cyclicity in convolutional codes.
\newblock {\em {IEEE} Transactions on Information Theory}, 62(5):2702--2706,
  2016.

\bibitem{GLN2017IdealCodes}
J.~G{\'{o}}mez{-}Torrecillas, F.~J. Lobillo, and G.~Navarro.
\newblock Ideal codes over separable ring extensions.
\newblock {\em {IEEE} Transactions on Information Theory}, 63(5):2796--2813,
  2017.
\newblock arXiv:1408.1546.

\bibitem{GLNPGZ}
J.~G\'{o}mez-Torrecillas, F.~J. Lobillo, and G.~Navarro.
\newblock {Peterson-Gorenstein-Zierler} algorithm for skew {RS} codes.
\newblock {\em Linear and Multilinear Algebra}, 2017.
\newblock See also https://arxiv.org/abs/1703.00745.

\bibitem{gln2017sugiyama}
J.~G{\'o}mez-Torrecillas, F.~J. Lobillo, and G.~Navarro.
\newblock A {Sugiyama}-like decoding algorithm for convolutional codes.
\newblock {\em {IEEE} Transactions on Information Theory}, In press, 2017.

\bibitem{Gorenstein/Zierler:1961}
D.~Gorenstein and N.~Zierler.
\newblock A class of error-correcting codes in \(p^m\) symbols.
\newblock {\em Journal of the Society for Industrial and Applied Mathematics},
  9(2):207--214, 1961.

\bibitem{Hartmann/Tzeng:1972}
C.~R.~P. Hartmann and K.~K. Tzeng.
\newblock Generalizations of the {BCH}-bound.
\newblock {\em Information and Control}, 20:489--498, 1972.

\bibitem{Hocquenghem:1959}
A.~Hocquenghem.
\newblock Codes correcteurs d'erreurs.
\newblock {\em Chiffres}, 2:147--156, 1959.

\bibitem{Huffman/Pless:2010}
W.~C. Huffman and V.~Pless.
\newblock {\em Fundamentals of Error-Correcting Codes}.
\newblock Cambridge University Press, 2010.

\bibitem{Jacobson:1996}
N.~Jacobson.
\newblock {\em {Finite-dimensional division algebras over fields.}}
\newblock {Berlin: Springer}, 1996.

\bibitem{Kshevetskiy/Gabidulin:2005}
A.~Kshevetskiy and E.~Gabidulin.
\newblock The new construction of rank codes.
\newblock In {\em Information Theory, 2005. ISIT 2005. Proceedings.
  International Symposium on}, pages 2105--2108. IEEE, 2005.

\bibitem{Lam/Leroy:1988}
T.~Y. Lam and A.~Leroy.
\newblock {Vandermonde} and {Wronskian} matrices over division rings.
\newblock {\em Journal of Algebra}, 119(2):308 -- 336, 1988.

\bibitem{Lam/Leroy/Ozturk:2008}
T.~Y. Lam, A.~Leroy, and A.~Ozturk.
\newblock Wedderburn polynomial over division rings, {II}.
\newblock {\em Contemp. Math.}, 456:73--98, 2008.

\bibitem{Liu/etal:2015}
S.~Liu, F.~Manganiello, and F.~R. Kschischang.
\newblock Construction and decoding of generalized skew-evaluation codes.
\newblock In {\em 2015 IEEE 14th CanadianWorkshop on Information Theory (CWIT),
  St. John's, NL}, pages 9--13, 2015.

\bibitem{MartinezPenas:2017}
U.~Mart{\'\i}nez-Pe{\~n}as.
\newblock On the roots and minimum rank distance of skew cyclic codes.
\newblock {\em Des. Codes Cryptogr.}, 83(3):639--660, 2017.

\bibitem{McDonald:1974}
B.~R. McDonald.
\newblock {\em Finite Rings With Identity}.
\newblock Marcel Dekker Inc, June 1974.

\bibitem{Ore:1933}
O.~Ore.
\newblock Theory of non-commutative polynomials.
\newblock {\em Annals of Mathematics}, 34:480--508, 1933.

\bibitem{Peterson:1960}
W.~W. Peterson.
\newblock Encoding and error-correction procedures for the bose-chaudhuri
  codes.
\newblock {\em {IRE} Transactions on Information Theory}, 6(4):459--470, 1960.

\bibitem{Piret:1975}
P.~Piret.
\newblock On a class of alternating cyclic convolutional codes.
\newblock {\em {IEEE} Transactions on Information Theory}, IT-21(1):64--69,
  1975.

\bibitem{Piret:1976}
P.~Piret.
\newblock Structure and constructions of cyclic convolutional codes.
\newblock {\em {IEEE} Transactions on Information Theory}, 22(2):147--155,
  1976.

\bibitem{Roos:1979}
C.~Roos.
\newblock On the structure of convolutional and cyclic convolutional codes.
\newblock {\em {IEEE} Transactions on Information Theory}, IT-25(6):676--683,
  1979.

\bibitem{sagemath}
{The Sage Developers}.
\newblock {\em {{S}ageMath, the {S}age {M}athematics {S}oftware {S}ystem
  ({V}ersion 7.6)}}.
\newblock The Sage Development Team, 2017.

\bibitem{vanderWaerden:1970}
B.~L. van~der Waerden.
\newblock {\em Algebra}, volume~1.
\newblock Frederick Ungar Publishing Co., 1970.

\bibitem{Lint/Wilson:1986}
J.~H. van Lint and R.~M. Wilson.
\newblock On the minimum distance of cyclic codes.
\newblock {\em IEEE Transactions on Information Theory}, 32(1):23--40, January
  1986.

\bibitem{Zeh/etal:2013}
A.~Zeh, A.~Wachter-Zeh, M.~Gadouleau, and S.~Bezzateev.
\newblock Generalizing bounds on the minimum distance of cyclic codes using
  cyclic product codes.
\newblock In {\em 2013 IEEE International Symposium on Information Theory},
  pages 126--130, July 2013.

\end{thebibliography}

\end{document}